\documentclass{amsart}
\usepackage{amsmath,amssymb,amscd,amsthm,latexsym,accents,stmaryrd}
\usepackage{amsfonts}
\usepackage[latin1]{inputenc}
\usepackage{graphicx}
\usepackage{float}
\usepackage{subfigure}
\usepackage[all]{xy}
\usepackage{enumerate}
\usepackage{amsaddr}
\usepackage{tikz}
\usetikzlibrary{automata,positioning}

\newtheorem{theor}{Theorem}[subsection]
\newtheorem{lem}[theor]{Lemma}
\newtheorem{prop}[theor]{Proposition}
\newtheorem{cor}[theor]{Corollary}
\theoremstyle{definition}
\newtheorem{defin}[theor]{Definition}
\theoremstyle{remark}
\newtheorem*{rem}{Remark}
\newtheorem*{rems}{Remarks}

\newtheorem*{ex}{Example}
\newtheorem*{exs}{Examples}

\newcommand{\iso}{\stackrel{\cong}{\rightarrow}}

\newdir{ >}{{}*!/-7pt/\dir{>}}

\def\R{{\mathbb{R}}}

\def\length{{\rm{length}}}
\def\A{{\mathcal{A}}}
\def\B{{\mathcal{B}}}

\def\deg{{\rm{deg}}}

\begin{document}

\title{Weak morphisms of higher dimensional automata}

\author{Thomas Kahl}

\address{Centro de Matem\'atica,
Universidade do Minho, Campus de Gualtar, \\
4710-057 Braga,
Portugal
}

\email{kahl@math.uminho.pt
}

\thanks{This research has been supported by FEDER funds through ``Programa Operacional Factores de Competitividade - COMPETE'' and by FCT -  \emph{Fundação para a Ciência e a Tecnologia} through projects Est-C/MAT/UI0013/2011 and PTDC/MAT/0938317/2008.}




\begin{abstract}
We introduce weak morphisms of higher dimensional automata and use them to define preorder relations for HDAs, among which homeomorphic abstraction and trace equivalent abstraction. It is shown that homeomorphic abstraction is essentially always stronger than trace equivalent abstraction. We also define the trace language of an HDA and show that, for a large class of HDAs, it is invariant under trace equivalent abstraction.
\end{abstract}

\maketitle

\section*{Introduction}

One of the most expressive models of concurrency is the one of higher dimensional automata \cite{vanGlabbeek}. A higher dimensional automaton (HDA) over a monoid $M$  is a precubical set with initial and final states and with 1-cubes labelled by elements of $M$ such that opposite edges of 2-cubes have the same label. Intuitively, an HDA can be seen as an automaton with an independence relation represented by cubes. If two actions $a$ and $b$ are enabled in a state $q$ and are independent in the sense that they may be executed in any order or even simultaneously without any observable difference, then the HDA  contains a 2-cube linking the two execution sequences $ab$ and $ba$ beginning in $q$. Similarly, the independence of $n$ actions is represented by $n$-cubes. It has been shown in \cite{vanGlabbeek} that many classical models of concurrency can be translated into the one of HDAs.  HDA semantics for process algebras are given in \cite{GaucherProcess} and \cite{GoubaultJensen}. 

In this paper, we introduce three preorder relations for HDAs. Whenever, as in figure \ref{weakmor}, an HDA $\B$ is a subdivision of an HDA $\A$, then $\A$ is related to $\B$ in each of these preorders. 
\begin{figure}[t]
\subfigure[$\A$]
{ 
\begin{tikzpicture}[initial text={},on grid]

\path[draw, fill=lightgray] (0,0)--(1,0)--(1,1)--(0,1)--cycle;

 \node[state,minimum size=0pt,inner sep =2pt,initial,fill=white] (q_0)   {}; 
    
   \node[state,minimum size=0pt,inner sep =2pt,fill=white] (q_2) [right=of q_0,xshift=0cm] {};
   
   \node[state,minimum size=0pt,inner sep =2pt,fill=white] [above=of q_0, yshift=0cm] (q_3)   {};

   \node[state,minimum size=0pt,inner sep =2pt,accepting,fill=white] (q_5) [right=of q_3,xshift=0cm] {}; 
   
    \path[->] 
    (q_0) edge[above] node {$ab$} (q_2)
    (q_3) edge[above]  node {$ab$} (q_5)
    (q_0) edge[left]  node {$c$} (q_3)
    (q_2) edge[right]  node {$c$} (q_5);

\end{tikzpicture}
}
\hspace{2cm}
\subfigure[$\B$]
{
\begin{tikzpicture}[initial text={},on grid]

\path[draw, fill=lightgray] (0,0)--(2,0)--(2,1)--(0,1)--cycle;

 \node[state,minimum size=0pt,inner sep =2pt,initial,fill=white] (q_0)   {}; 
    
   \node[state,minimum size=0pt,inner sep =2pt,fill=white] (q_1) [right=of q_0,xshift=0cm] {};

\node[state,minimum size=0pt,inner sep =2pt,fill=white] (q_2) [right=of q_1,xshift=0cm] {};   
   
   \node[state,minimum size=0pt,inner sep =2pt,fill=white] [above=of q_0, yshift=0cm] (q_3)   {};
   
   \node[state,minimum size=0pt,inner sep =2pt,fill=white] (q_4) [right=of q_3,xshift=0cm] {};
   
   \node[state,minimum size=0pt,inner sep =2pt,accepting,fill=white] (q_5) [right=of q_4,xshift=0cm] {};

    \path[->] 
    (q_0) edge[above]  node {$a$} (q_1)
    (q_1) edge[above]  node {$b$} (q_2)
    (q_3) edge[above]  node {$a$} (q_4)
    (q_4) edge[above]  node {$b$} (q_5)
    (q_0) edge[left]  node {$c$} (q_3)
    (q_2) edge[right]  node {$c$} (q_5)
    (q_1) edge[left]  node {$c$} (q_4);

\end{tikzpicture}
}
\caption{Two HDAs $\A$ and $\B$ over the free monoid on $\{a,b,c\}$ such that $\B$ is a subdivision of $\A$, $\A \to \B$, $\A \stackrel{\sim_t}{\to} \B$ and $\A \stackrel{\approx}{\to} \B$} \label{weakmor}
\end{figure} 
The definitions of the preorder relations are based on the concept of weak morphism, which is developed in section \ref{weakmorphisms}. Roughly speaking, a weak morphism between two HDAs is a continuous map between their geometric realisations that sends subdivided cubes to subdivided cubes and that preserves labels of paths. A morphism of HDAs, or, more precisely, its geometric realisation, is a weak morphism but not vice versa. For example, in figure \ref{weakmor}, there exists a weak morphism from $\A$ to $\B$, but there does not exist any morphism between the two HDAs. If there exists a weak morphism from an HDA $\A$ to an HDA $\B$, we write $\A \to \B$. The relation $\to$ is the first preorder relation for HDAs we consider in this paper. It has the basic property that $\A \to \B$ implies that the language (the behaviour) of $\A$ is contained in the language of $\B$. One may consider $\to$ as a kind of simulation preorder: If there exists a weak morphism from $\A$ to $\B$, then for every 1-cube in $\A$ with label $\alpha$ from a vertex $v$ to a vertex $w$ there exists a path in $\B$ with label $\alpha$ from the image of $v$ to the image of $w$.

The higher dimensional structure of an HDA induces an independence relation on the monoid of labels and allows us to define the trace language of an HDA in section \ref{HDA}. The fundamental concept in this context is dihomotopy (short for directed homotopy) of paths \cite{FajstrupGR, Goubault}. Besides the trace language of an HDA, we also consider the trace category of an HDA and trace equivalences between HDAs. The trace category of an HDA is a variant of the fundamental bipartite graph of a d-space \cite{BubenikExtremal}. Its objects are certain important states of the HDA, including the initial and the final ones, and the morphisms are the dihomotopy classes of paths between these states. A trace equivalence is essentially defined as a weak morphism that induces an isomorphism of trace categories. If there exists a trace equivalence from an HDA $\A$ to an HDA $\B$, we say that $\A$ is a trace equivalent abstraction of $\B$ and write $\A \stackrel{\sim_t}{\to} \B$. Trace equivalent abstraction is our second preorder relation for HDAs. We show that for a large class of HDAs, $\A \stackrel{\sim_t}{\to} \B$ implies that there exists a bijection between the trace languages of $\A$ and $\B$. 

The third preorder relation is called homeomorphic abstraction and is the subject of section \ref{homeoabs}. We say that an HDA $\A$ is a homeomorphic abstraction of an HDA $\B$ and write $\A \stackrel{\approx}{\to} \B$ if there exists a weak morphism from $\A$ to $\B$ that is a homeomorphism and a bijection on initial and on final states. Homeomorphic abstraction may be seen as a labelled version of T-homotopy equivalence in the sense of \cite{GaucherGoubault}. We show that under a mild condition, $\A \stackrel{\approx}{\to} \B$ implies $\A \stackrel{\sim_t}{\to} \B$.

\section{Precubical sets and HDAs}

This section contains some basic and well-known material on precubical sets and higher dimensional automata. 

\subsection{Precubical sets}

A \emph{precubical set} is a graded set $P = (P_n)_{n \geq 0}$ with  \emph{boundary operators} $d^k_i: P_n \to P_{n-1}$ $(n>0,\;k= 0,1,\; i = 1, \dots, n)$ satisfying the relations $d^k_i\circ d^l_{j}= d^l_{j-1}\circ d^k_i$ $(k,l = 0,1,\; i<j)$ \cite{FahrenbergThesis, Fajstrup, FajstrupGR, GaucherGoubault, Goubault}. The least $n\geq 0$ such that $P_i = \emptyset$ for all $i>n$ is called the \emph{dimension} of $P$. If no such $n$ exists, then the dimension of $P$ is $\infty$. If $x\in P_n$, we say that $x$ is of \emph{degree} $n$ and write $\deg(x) = n$. The elements of degree $n$ are called the \emph{$n$-cubes} of $P$. The elements of degree $0$ are also called the \emph{vertices} or the \emph{nodes} of $P$. A morphism of precubical sets is a morphism of graded sets that is compatible with the boundary operators. 

The category of precubical sets can be seen as the presheaf category of functors $\Box^{\mbox{\tiny op}} \to \bf Set$ where $\Box$ is the small subcategory of the category of topological spaces whose objects are the standard $n$-cubes $[0,1]^n$ $(n \geq 0)$ and whose non-identity morphisms are composites of the maps $\delta^k_i\colon [0,1]^n\to [0,1]^{n+1}$ ($k \in \{0,1\}$, $n \geq 0$, $i \in  \{1, \dots, n+1\}$) given by $\delta_i^k(u_1,\dots, u_n)= (u_1,\dots, u_{i-1},k,u_i \dots, u_n)$. Here, we use the convention that given a topological space $X$, $X^0$ denotes the one-point space $\{()\}$.

\subsection{Precubical subsets}

A \emph{precubical subset} of a precubical set $P$ is a graded subset of $P$ that is stable under the boundary  operators. It is clear that a precubical subset is itself a precubical set. Note that unions and intersections of precubical subsets are precubical subsets and that the image of a morphism  $f\colon P \to Q$ of precubical sets is a precubical subset of $Q$. 

\subsection{Intervals}

Let $k$ and $l$ be integers such that $k \leq l$. The \emph{precubical interval}  $\llbracket k,l \rrbracket$ is the at most $1$-dimensional precubical set defined by $\llbracket k,l \rrbracket_0 = \{k,\dots , l\}$, $\llbracket k,l \rrbracket_1 = \{[k,k+1], \dots , [l-1,l]\}$, $d_1^0[j-1,j] = j-1$ and $d_1^1[j-1,j] = j$.

\subsection{Tensor product} 

Given two graded sets $P$ and $Q$, the \emph{tensor product} $P\otimes Q$ is the graded set defined by $(P\otimes Q)_n = \coprod \limits_{p+q = n} P_p\times Q_q$. If $P$ and $Q$ are precubical sets, then $P\otimes Q$ is a precubical set with respect to the boundary operators given by
$$d_i^k(x,y) = \left\{ \begin{array}{ll} (d_i^kx,y), & 1\leq i \leq \deg(x),\\
(x,d_{i-\deg(x)}^ky), & \deg(x)+1 \leq i \leq \deg(x) + \deg(y)
\end{array}\right.$$ 
(cf. \cite{FahrenbergThesis}). The tensor product turns the categories of graded and precubical sets into monoidal categories.  

The $n$-fold tensor product of a graded or precubical set $P$ is denoted by $P^{\otimes n}$. Here, we use the convention $P^{\otimes 0} = \llbracket 0,0\rrbracket = \{0\}$. The \emph{precubical $n$-cube} is the precubical set $\llbracket 0,1\rrbracket^ {\otimes n}$. The only element of degree $n$ in $\llbracket 0,1\rrbracket^ {\otimes n}$ will be denoted by $\iota_n$. We thus have $\iota_0 = 0$ and $\iota_n = (\underbrace{ [0,1] ,\dots , [ 0,1]}_{n\; { \rm{ times}}})$ for $n>0$.

\subsection{The morphism corresponding to an element} 

Let $x$ be an element of degree $n$ of a precubical set $P$. Then there exists a unique morphism of precubical sets $x_{\sharp}\colon \llbracket 0,1\rrbracket ^{\otimes n}\to P$ such that $x_{\sharp}(\iota_n) = x$. Indeed, 
by the Yoneda lemma, there exist unique morphisms of precubical sets $f\colon \Box(-,[0,1]^n) \to P$ and $g\colon \Box(-,[0,1]^n) \to \llbracket 0,1\rrbracket ^{\otimes n}$ such that $f(id_{[0,1]^n}) = x$ and $g(id_{[0,1]^n}) = \iota_n$. The map $g$ is an isomorphism, and $x_{\sharp} = f\circ g^{-1}$. 
\subsection{Paths} 

A \emph{path of length $k$} in a precubical set $P$ is a morphism of precubical sets $\omega \colon \llbracket 0,k \rrbracket \to P$. The set of paths in $P$ is denoted by $P^{\mathbb I}$. If $\omega \in P^{\mathbb I}$ is a path of length $k$, we write $\length(\omega) = k$. The \emph{concatenation} of two paths $\omega \colon \llbracket 0,k \rrbracket \to P$ and $\nu \colon \llbracket 0,l \rrbracket \to P$ with $\omega (k) = \nu (0)$ is the path $\omega \cdot \nu\colon \llbracket 0,{k+l} \rrbracket \to P$ defined by 
$$\omega \cdot \nu (j) = \left\{\begin{array}{ll} \omega (j), & 0\leq j \leq k,\\ \nu(j-k), & k \leq j \leq k+l \end{array}\right.$$ and $$\omega \cdot \nu ([j-1,j]) = \left\{\begin{array}{ll} \omega ([j-1,j]) & 0< j \leq k,\\ \nu([j-k-1,j-k]) & k < j \leq k+l. \end{array}\right. $$ Clearly, concatenation is associative. Note that for any path $\omega \in P^{\mathbb I}$ of length $k \geq 1$ there exists a unique sequence $(x_1, \dots , x_k)$ of elements of $P_1$ such that $d_1^0x_{j+1} = d_1^1x_j$ for all $1\leq j < k$ and $\omega = x_{1\sharp} \cdots x_{k\sharp}$.

\subsection{Geometric realisation}
 \label{geomreal}  

The \emph{geometric realisation} of a precubical set $P$ is the quotient space $|P|=(\coprod _{n \geq 0} P_n \times [0,1]^n)/\sim$ where the sets $P_n$ are given the discrete topology and the equivalence relation is given by
$$(d^k_ix,u) \sim (x,\delta_i^k(u)), \quad  x \in P_{n+1},\; u\in [0,1]^n,\; i \in  \{1, \dots, n+1\},\; k \in \{0,1\}$$ (see \cite{FahrenbergThesis}, \cite{Fajstrup}, \cite{FajstrupGR},  \cite{GaucherGoubault}, \cite{Goubault}). The geometric realisation of a morphism of precubical sets $f\colon P \to Q$ is the continuous map $|f|\colon |P| \to |Q|$ given by $|f|([x,u])= [f(x),u]$. We remark that the geometric realisation is a functor from the category of precubical sets to the category of topological spaces.

\begin{exs}
(i) The geometric realisation of the precubical $n$-cube can be identified with the standard $n$-cube by means of the homeomorphism $[0,1]^n \to |\llbracket 0,1\rrbracket^ {\otimes n}|$, $u \mapsto [\iota_n,u]$.

(ii) The geometric realisation of the precubical interval $\llbracket k,l \rrbracket$ can be identified with the closed interval $[k,l]$ by means of the homeomorphism $|\llbracket k,l \rrbracket| \to [k,l]$ given by $[j,()] \mapsto j$ and $[[j-1,j],t] \mapsto j-1+t$. Using this correspondence, the geometric realisation of a precubical path $\llbracket 0, k\rrbracket \to P$ can be seen as a path $[0,k] \to |P|$, and under this identification we have that $|\omega \cdot \nu| = |\omega|\cdot |\nu|$. 
\end{exs}

We note that for every element $a \in |P|$ there exist a unique integer $n\geq 0$, a unique element $x \in P_n$ and a unique element $u \in ]0,1[^ n$ such that $a = [x,u]$.

The geometric realisation of a precubical set $P$ is a CW-complex \cite{GaucherGoubault}. The $n$-skeleton of $|P|$ is the geometric realisation of the precubical subset $P_{\leq n}$ of $P$ defined by $(P_{\leq n})_m = P_m$ $(m\leq n)$. The closed $n$-cells of $|P|$ are the spaces $|x_{\sharp}(\llbracket 0,1\rrbracket^{\otimes n})|$ where $x \in P_n$. The characteristic map of the cell $|x_{\sharp}(\llbracket 0,1\rrbracket^{\otimes n})|$ is the map $[0,1]^n \stackrel{\approx}{\to} |\llbracket 0,1\rrbracket ^{\otimes n}| \stackrel{|x_{\sharp}|}{\to} |P|, u \mapsto [x,u]$. The geometric realisation of a precubical subset $Q$ of $P$ is a subcomplex of $|P|$.

The geometric realisation is a comonoidal functor with respect to the natural continuous map $\psi_{P,Q}\colon |P\otimes Q| \to |P| \times |Q|$ given by \begin{eqnarray*}\lefteqn{\psi_{P,Q} ([(x,y),(u_1,\dots ,u_{\deg(x)+\deg(y)})])}\\ &=& ([x,(u_1,\dots , u_{\deg(x)})],[y,(u_{\deg(x)+1}, \dots u_{\deg(x)+\deg(y)}]).\end{eqnarray*} If $P$ and $Q$ are finite, then $\psi_{P,Q}$ is a homeomorphism and permits us to identify $|P\otimes Q|$ with $|P|\times |Q|$. We may thus identify the geometric realisation of a precubical set of the form $\llbracket k_1,l_1\rrbracket \otimes \cdots \otimes \llbracket k_n,l_n\rrbracket$ $(k_i<l_i)$ with the product $[k_1,l_1]\times \cdots \times [k_n,l_n]$ by means of the correspondence $$\left[([i_1,i_1+1], \dots,[i_n,i_n+1]),(u_1,\dots,u_n)\right] \mapsto (i_1+u_1,\dots ,i_n+u_n).$$

\subsection{Higher dimensional automata} \label{HDAdef}

Let $M$ be a monoid. A \emph{higher dimensional automaton over} $M$ (abbreviated $M$-HDA or simply HDA) is a tuple $\A = (P,I,F, \lambda)$ where $P$ is a precubical set, $I \subseteq P_0$ is a set of \emph{initial states}, $F \subseteq P_0$ is a set of \emph{final states}, and $\lambda \colon P_1 \to M$ is a map, called the \emph{labelling function}, such that $\lambda (d_i^0x) = \lambda (d_i^1x)$ for all $x \in P_2$ and $i \in \{1,2\}$. A \emph{morphism} from an $M$-HDA $\A = (P,I, F,\lambda)$ to an $M$-HDA $\B = (P',I', F',\lambda')$ is a morphism of precubical sets  $f\colon P \to P'$  such that $f(I) \subseteq I'$, $f(F) \subseteq F'$ and  $\lambda'(f(x)) = \lambda(x)$ for all $x \in P_1$.

This definition of higher dimensional automata is essentially the same as the one in \cite{vanGlabbeek}. Besides the fact that we consider a monoid and not just a set of labels, the only difference is that in \cite{vanGlabbeek} an HDA is supposed to have exactly one initial state. Note that 1-dimensional $M$-HDAs and morphisms of 1-dimensional $M$-HDAs are the same as automata over $M$ and automata morphisms as defined in \cite{Sakarovitch}.

\subsection{The language accepted by an HDA} \label{LA}
 
Let $\A = (P,I,F, \lambda)$ be an $M$-HDA. The \emph{extended labelling function} of $\A$ is the map $\overline{\lambda} \colon P^{\mathbb I} \to M$ defined as follows: If $\omega = x_{1\sharp} \cdots x_{k\sharp}$ for a sequence $(x_1, \dots , x_k)$ of elements of $P_1$ such that $d_1^0x_{j+1} = d_1^1x_j$ $(1\leq j < k)$, then we set $\overline{\lambda} (\omega) = \lambda (x_1) \cdots \lambda (x_k)$; if $\omega$ is a path of length $0$, then we set  $\overline{\lambda} (\omega ) = 1$. The \emph{language accepted by} $\A$ is the set $$L(\A) = \{\overline{\lambda} (\omega) : \omega \in P^{\mathbb I},\; \omega (0) \in I,\; \omega(\length(\omega)) \in F\}.$$ Note that $L(\A)$ is the the behaviour in the sense of \cite{Sakarovitch} of the 1-skeleton of $\A$, i.e. the $M$-HDA $(P_{\leq 1},I,F, \lambda)$. A fundamental property of the language accepted by an HDA is that one has $L(\A) \subseteq L(\B)$ if there exists a morphism of $M$-HDAs from $\A$ to $\B$ \cite[prpty. II.3.1]{Sakarovitch}.

Recall, for instance from \cite{Diekert} or \cite{Sakarovitch}, that a subset $L$ of a monoid $M$ is called 
\begin{itemize}
\item 
\emph{rational} if it belongs to the smallest subset of the power set $\mathfrak{P}(M)$ that contains all finite subsets of $M$ and is closed under the operations union, product and star; \item \emph{recognisable} if there exist a right action of $M$ on a finite set $S$, an element $s_0\in S$ and a subset $T \subseteq S$ such that $L = \{m\in M: s_0\cdot m \in T\}$. 
\end{itemize} A subset of a monoid $M$ is rational if and only if it is the languague accepted by an $M$-HDA $\A = (P,I,F, \lambda)$ such that $P_0$ and $P_1$ are finite \cite[thm. II.1.1]{Sakarovitch}. By Kleene's theorem, a subset of a finitely generated free monoid is rational if and only if it is recognisable. Consequently, if $M$ is a finitely generated free monoid and $\A = (P,I,F, \lambda)$ is an $M$-HDA such that $P_0$ and $P_1$ are finite, then $L(\A)$ is both a rational and a recognisable subset of $M$.

\section{Weak morphisms} \label{weakmorphisms}

In this section, we introduce weak morphisms of precubical sets and HDAs and use them to define the preorder relation $\to$. We establish the basic properties of weak morphisms and show, in particular, that a weak morphism induces a map of path sets. We also define subdivisions, which provide an important example of weak morphisms. We begin by studying dihomeomorphisms of hyperrectangles, which are central to the concept of weak morphism. 

\subsection{Dihomeomorphisms of hyperrectangles} Let  $$\phi\colon [a_1,b_1]\times \cdots \times [a_n,b_n] \to [c_1,d_1] \times \cdots \times [c_n,d_n] \quad (a_i < b_i,\, c_i < d_i)$$ be a homeomorphism such that $\phi$ and $\phi^{-1}$ preserve the natural partial order of $\R^n$. Such a homeomorphism is called a \emph{dihomeomorphism} (see e.g. \cite{FajstrupGR}).

\begin{prop} \label{minmaxcoord}
Consider an element $x\in [a_1,b_1]\times \cdots \times [a_n,b_n]$. Then $x$ and $\phi(x)$ have the same number of minimal coordinates and the same number of maximal coordinates.
\end{prop}

\begin{proof}
Consider the sets $$A = \{u\in [a_1,b_1]\times \cdots \times [a_n,b_n]\colon u \leq x\} = [a_1,x_1] \times \cdots \times [a_n,x_n]$$
and $$B = \{v\in [c_1,d_1] \times \cdots \times [c_n,d_n]\colon v \leq \phi(x)\} = [c_1,\phi_1(x)] \times \cdots \times [c_n,\phi_n(x)].$$
Then $\phi$ restricts to a homeomorphism from $A$ to $B$. It follows that $A$ and $B$ are spaces of the same dimension. Hence $x$ and $\phi(x)$ have the same number of minimal coordinates. In the same way one shows that $x$ and $\phi(x)$ have the same number of maximal coordinates.
\end{proof}

\begin{prop} \label{orderhomeo}
Consider an index $k \in \{1, \dots, n\}$. Then there exists an index $l\in \{1, \dots, n\}$ such that \begin{eqnarray*}\lefteqn{\phi ([a_1,b_1]\times \cdots \times [a_{k-1},b_{k-1}]\times \{a_k\} \times [a_{k+1},b_{k+1}]\times \cdots \times [a_n,b_n])}\\ &=& [c_1,d_1]\times \cdots \times [c_{l-1},d_{l-1}]\times \{c_l\} \times [c_{l+1},d_{l+1}]\times \cdots \times [c_n,d_n]\end{eqnarray*} and \begin{eqnarray*}\lefteqn{\phi ([a_1,b_1]\times \cdots \times [a_{k-1},b_{k-1}]\times \{b_k\} \times [a_{k+1},b_{k+1}]\times \cdots \times [a_n,b_n])}\\ &=& [c_1,d_1]\times \cdots \times [c_{l-1},d_{l-1}]\times \{d_l\} \times [c_{l+1},d_{l+1}]\times \cdots \times [c_n,d_n].\end{eqnarray*}
\end{prop}

\begin{proof}
Consider the point $(b_1, \dots ,b_{k-1},a_k,b_{k+1}, \dots,b_n) \in [a_1,b_1]\times \cdots \times [a_n,b_n]$. By \ref{minmaxcoord}, there exists an index $l\in \{1, \dots, n\}$ such that $$\phi(b_1, \dots ,b_{k-1},a_k,b_{k+1}, \dots,b_n) = (d_1, \dots ,d_{l-1},c_l,d_{l+1}, \dots,d_n).$$ It follows that \begin{eqnarray*}
\lefteqn{\phi ([a_1,b_1]\times \cdots \times [a_{k-1},b_{k-1}]\times \{a_k\} \times [a_{k+1},b_{k+1}]\times \cdots \times [a_n,b_n])}\\ 
&=& \phi (\{x\in [a_1,b_1]\times \cdots \times [a_n,b_n]\colon x \leq (b_1, \dots ,b_{k-1},a_k,b_{k+1}, \dots,b_n)\})\\
&=& \{y\in [c_1,d_1]\times \cdots \times [c_n,d_n]\colon y \leq (d_1, \dots ,d_{l-1},c_l,d_{l+1}, \dots, d_n)\}\\
&=& [c_1,d_1]\times \cdots \times [c_{l-1},d_{l-1}]\times \{c_l\} \times [c_{l+1},d_{l+1}]\times \cdots \times [c_n,d_n]
\end{eqnarray*}
In the same way, one shows that there exists an index $p\in \{1, \dots, n\}$ such that \begin{eqnarray*}\lefteqn{\phi ([a_1,b_1]\times \cdots \times [a_{k-1},b_{k-1}]\times \{b_k\} \times [a_{k+1},b_{k+1}]\times \cdots \times [a_n,b_n])}\\ &=& [c_1,d_1]\times \cdots \times [c_{p-1},d_{p-1}]\times \{d_p\} \times [c_{p+1},d_{p+1}]\times \cdots \times [c_n,d_n].\end{eqnarray*}
It remains to show that $l=p$. Suppose that this is not the case. Consider the points $x_{\theta} = (\frac{a_1+b_1}{2}, \dots , \frac{a_{k-1}+b_{k-1}}{2},a_k+\theta (b_k-a_k),\frac{a_{k+1}+b_{k+1}}{2}, \dots, \frac{a_n+b_n}{2})$, $\theta \in [0,1]$. Write $\phi(x_0) = (s_1, \dots , s_{l-1},c_l,s_{l+1}, \dots , s_n)$ and $\phi(x_1) = (t_1, \dots , t_{p-1},d_p,t_{p+1}, \dots , t_n)$. Since $x_0 \leq x_1$, we have $s_j \leq t_j$ for all $j \in \{1, \dots,n\} \setminus \{l,p\}$. Consider the element $v$ of  $[c_1,d_1] \times \cdots \times [c_n,d_n]$ given by $$v_j = \left\{\begin{array}{ll}c_l, & j=l,\\ d_p, &j=p,\\ s_j,& j\in \{1,\dots n\}\setminus\{l,p\}.\end{array}\right.$$ Then $\phi(x_0) \leq v \leq \phi(x_1)$ and hence $x_0 \leq \phi^{-1}(v) \leq x_1$. It follows that there exists $\theta \in [0,1]$ such that $\phi^{-1}(v) = x_{\theta}$, i.e. $\phi(x_{\theta}) = v$. By \ref{minmaxcoord}, $x_{\theta}$ has at least one minimal and one maximal coordinate. This is impossible, and therefore $l=p$.
\end{proof}
 
\subsection{Weak morphisms of precubical sets} \label{defweakmor}
A \emph{weak morphism} from a precubical set $P$ to a precubical set $P'$ is a continuous map $f\colon |P| \to |P'|$ such that the following two conditions hold:
\begin{enumerate}
\item for every vertex $v\in P_0$ there exists a (necessarily unique) vertex $f_0(v)\in P'_0$ such that $f([v,()]) = [f_0(v),()]$;
\item for all integers $n, k_1, \dots, k_n\geq 1$ and every morphism of precubical sets $\chi \colon \llbracket 0,{k_1} \rrbracket\otimes \cdots \otimes \llbracket 0,{k_n} \rrbracket \to P$ there exist integers $k'_1, \dots, k'_n\geq 1$, a morphism of precubical sets $\chi' \colon \llbracket 0,{k'_1} \rrbracket\otimes \cdots \otimes \llbracket 0,{k'_n} \rrbracket \to P'$ and a dihomeo\-morphism 
\begin{eqnarray*}\lefteqn{\phi\colon  |\llbracket 0,{k_1} \rrbracket\otimes \cdots \otimes \llbracket 0,{k_n} \rrbracket| = [0,k_1] \times \cdots \times [0,k_n]}\\ &\to& |\llbracket 0,{k'_1} \rrbracket\otimes \cdots \otimes \llbracket 0,{k'_n} \rrbracket|= [0,k'_1] \times \cdots \times [0,k'_n]\end{eqnarray*} such that $f\circ |\chi| = |\chi'|\circ \phi$.
\end{enumerate}

\begin{exs}
(i) The geometric realisation of a morphism of precubical sets is a weak morphism.  

(ii) The map $|\llbracket 0,1 \rrbracket^ {\otimes 2}| = [0,1]^2 \to |\llbracket 0,1 \rrbracket^ {\otimes 2}| = [0,1]^2$, $(s,t) \mapsto (t,s)$ is a weak morphism that is not the geometric realization of a morphism of precubical sets  $\llbracket 0,1 \rrbracket^ {\otimes 2} \to \llbracket 0,1 \rrbracket^ {\otimes 2}$.

(iii) The map $|\llbracket 0,1 \rrbracket| = [0,1] \to |\llbracket 0,2 \rrbracket| = [0,2]$, $t \mapsto 2t$ is a weak morphism. This weak morphism is the homeomorphism of a subdivision in the sense of \ref{defsubdiv}.
\end{exs}

It is important to note that weak morphisms are stable under composition. This fact will enable us in sections \ref{defHDAmor}, \ref{simt} and \ref{homeoref} to use weak morphisms to define preorder relations for higher dimensional automata.

\begin{prop} \label{composition}
Let $f\colon |P| \to |Q|$ and $g\colon |Q| \to |R|$ be weak morphisms of precubical sets. Then $g\circ f\colon |P| \to |R|$ is a weak morphism of precubical sets, and $(g\circ f)_0 = g_0\circ f_0$. \hfill $\square$
\end{prop}

\subsection{The maps $\chi'$ and $\phi$} Our purpose in this subsection is to show that the morphism of precubical sets $\chi'$ and the dihomeomorphism $\phi$ in condition (2) of the definition of weak morphisms in section \ref{defweakmor} are unique and that $\phi$ is itself a weak morphism. We need four lemmas.

\begin{lem} \label{reallem}
Let $f\colon P \to R$ and $g\colon Q \to R$ be  morphisms of precubical sets and $\alpha \colon |P| \to |Q|$ be a continuous map such that $|g|\circ \alpha  = |f|$. Then there exists a morphism of precubical sets $h \colon P\to Q$ such that $g\circ h = f$ and $|h| = \alpha$.
\end{lem}

\begin{proof}
Let $x \in P$ be an element of degree $r$ and $u \in ]0,1[^r$. Consider the uniquely determined elements $y \in Q$ and $v \in ]0,1[^{\deg(y)}$ such that $\alpha([x,u]) = [y,v]$. We have $[g(y),v] = |g|([y,v]) = |g|\circ \alpha ([x,u])  = |f|([x,u]) = [f(x),u]$ and therefore $\deg(y) = r$, $g(y) = f(x)$ and $v=u$. Consider $u' \in ]0,1[^r$. For each $t \in [0,1]$, there exists an element $y_t \in Q_r$ such that $\alpha([x,(1-t)u+tu']) = [y_t,(1-t)u+tu']$. Since every continuous path in the subspace of open $r$-cells of $|Q|$ must stay in one of the cells, we have $y_t = y$ for all $t\in [0,1]$. Set $h(x) =y$. We have shown that $g(h(x)) = f(x)$ and that $\alpha ([x,u]) = [h(x),u]$ for all $u \in ]0,1[^{r}$. It remains to show that $h$ is a morphism of precubical sets. Consider an element $x \in P$ of degree $r>0$, $i \in \{1, \dots , r\}$ and $k \in \{0,1\}$. Consider the path $[0,1] \to |P|$, $t \mapsto [x,u_t]$ where the element $u_t \in [0,1]^r$ is defined by $$(u_t)_j = \left\{\begin{array}{ll} \vspace{0.2cm} \tfrac{1}{2}, & j\not=i,\\
\tfrac{1-t}{2} +tk,& j=i. \end{array}\right.$$
For $t\in [0,1[$, $u_t \in ]0,1[^r$ and therefore $\alpha ([x,u_t]) = [h(x),u_t]$. Hence also $\alpha ([x,u_1])  = [h(x),u_1]$. Thus,  $[h(d_i^kx),(\frac{1}{2}, \dots ,\frac{1}{2})] = \alpha ([d_i^kx,(\frac{1}{2}, \dots ,\frac{1}{2})]) = \alpha ([x,\delta_i^k(\frac{1}{2}, \dots ,\frac{1}{2})]) = \alpha ([x,u_1]) = [h(x),u_1] = [h(x),\delta_i^k(\frac{1}{2}, \dots ,\frac{1}{2})] = [d_i^kh(x),(\frac{1}{2}, \dots ,\frac{1}{2})]$. It follows that $h(d_i^kx) = d_i^kh(x)$ and hence that $h$ is a morphism of precubical sets. 
\end{proof}

\begin{lem} \label{incl}
Consider integers $n, k_1, \dots ,k_n \geq 1$ and a morphism of precubical sets $f\colon \llbracket 0,{k_1} \rrbracket\otimes \cdots \otimes \llbracket 0,{k_n} \rrbracket \to \llbracket 0,{l_1} \rrbracket\otimes \cdots \otimes \llbracket 0,{l_n} \rrbracket$ such that $f(0,\dots, 0) = (0, \dots ,0)$. Then $k_i \leq l_i$ for all $i \in \{1, \dots, n\}$ and $f$ is the inclusion.
\end{lem}

\begin{proof}
Note first that necessarily $l_1, \dots l_n \geq 1$. Consider the set $$K= \{0,\dots,k_1-1\}\times \cdots \times \{0,\dots,k_n-1\}.$$ For $(i_1,\dots,i_n) \in K$ set $c_{i_1,\dots,i_n} = ([i_1,i_1+1], \dots, [i_n,i_n+1]).$ The $c_{i_1,\dots,i_n}$ are the $n$-cubes of $\llbracket 0,{k_1} \rrbracket\otimes \cdots \otimes \llbracket 0,{k_n} \rrbracket$. In the same way, we use the notation $c'_{i_1,\dots,i_n}$ for the $n$-cubes of $\llbracket 0,{l_1} \rrbracket\otimes \cdots \otimes \llbracket 0,{l_n} \rrbracket$. Consider the total order on $K$ defined by $$(j_1,\dots,j_n) < (i_1,\dots,i_n)\Leftrightarrow \exists r \in \{1, \dots, n\}\colon j_1=i_1, \dots, j_{r-1}=i_{r-1},j_r<i_r.$$  We show by induction that for each $(i_1,\dots,i_n) \in K$, $i_1< l_1, \dots , i_n< l_n$ and $f(c_{i_1,\dots,i_n}) = c'_{i_1, \dots,i_n}$. The only $n$-cube of $\llbracket 0,{l_1} \rrbracket\otimes \cdots \otimes \llbracket 0,{l_n} \rrbracket$ having $(0,\dots, 0)$ as a vertex is $c'_{0,\dots,0}$. Therefore $f(c_{0,\dots,0}) = c'_{0,\dots,0}$. Suppose that $(i_1,\dots,i_n) > (0,\dots, 0)$ and that the assertion holds all $(j_1,\dots,j_n) < (i_1,\dots,i_n)$. Let $r$ be an index such that $i_r>0$. By the inductive hypothesis, $f(c_{i_1, \dots,i_{r-1},i_r-1,i_{r+1}, \dots,i_n}) = c'_{i_1, \dots,i_{r-1},i_r-1,i_{r+1}, \dots,i_n}$. It follows that \begin{eqnarray*}\lefteqn{d_r^0f(c_{i_1, \dots, i_n})}\\&=& f(d_r^0c_{i_1, \dots, i_n})\\ &=& f(d_r^1c_{i_1, \dots,i_{r-1},i_r-1,i_{r+1}, \dots,i_n})\\ &=& d_r^1f(c_{i_1, \dots,i_{r-1},i_r-1,i_{r+1}, \dots,i_n})\\ &=& d_r^1c'_{i_1, \dots,i_{r-1},i_r-1,i_{r+1}, \dots,i_n}\\ &=& ([i_1,i_1+1], \dots,[i_{r-1},i_{r-1}+1],i_r, [i_{r+1},i_{r+1}+1],\dots,[i_n,i_n+1]) \end{eqnarray*}
This can only happen if $i_1< l_1, \dots , i_n< l_n$ and $f(c_{i_1,\dots,i_n}) = c'_{i_1, \dots,i_n}$. The result follows.
\end{proof}

\begin{lem}\label{uniqueness}
Consider integers $n, k_1, \dots, k_n,l_1, \dots, l_n\geq 1$, morphisms of precubical sets $\xi \colon \llbracket 0,{k_1} \rrbracket\otimes \cdots \otimes \llbracket 0,{k_n} \rrbracket \to P$ and $\zeta \colon \llbracket 0,{l_1} \rrbracket\otimes \cdots \otimes \llbracket 0,{l_n} \rrbracket \to P$ and a homeomorphism $\alpha \colon |\llbracket 0,{k_1} \rrbracket\otimes \cdots \otimes \llbracket 0,{k_n} \rrbracket| \to |\llbracket 0,{l_1} \rrbracket\otimes \cdots \otimes \llbracket 0,{l_n} \rrbracket|$ such that $|\zeta| \circ \alpha  = |\xi|$. Then $k_i = l_i$ for all $i \in \{1, \dots,n\}$, $\alpha = id$ and $\xi = \zeta$.   
\end{lem}

\begin{proof}
By \ref{reallem}, there exist morphisms of precubical sets $h \colon \llbracket 0,{k_1} \rrbracket\otimes \cdots \otimes \llbracket 0,{k_n} \rrbracket \to \llbracket 0,{l_1} \rrbracket\otimes \cdots \otimes \llbracket 0,{l_n} \rrbracket$ and $h' \colon \llbracket 0,{l_1} \rrbracket\otimes \cdots \otimes \llbracket 0,{l_n} \rrbracket \to \llbracket 0,{k_1} \rrbracket\otimes \cdots \otimes \llbracket 0,{k_n} \rrbracket$ such that $|h| = \alpha$, $|h'| = \alpha^{-1}$, $\zeta \circ h = \xi$ and $\xi \circ h'= \zeta$. We show that $h (0, \dots ,0) = (0, \dots, 0)$. Suppose that this is not case. Then $h (0, \dots ,0) = d_1^1x$ for some $x \in \left(\llbracket 0,{l_1} \rrbracket\otimes \cdots \otimes \llbracket 0,{l_n} \rrbracket\right)_1$. Therefore $h'\circ h (0,\dots,0) = h '(d_1^ 1x) = d_1^ 1h'(x)$. On the other hand, $[h' \circ h (0,\dots,0),()] = |h' \circ h|([(0,\dots,0),()]) = \alpha^ {-1}\circ \alpha ([(0,\dots,0),()]) = [(0,\dots,0),()]$ and hence $h' \circ \nolinebreak h (0,\dots,0) = (0,\dots, 0)$. Thus, $(0,\dots,0) = d_1^ 1h'(x)$. This is impossible. It follows that $h (0, \dots ,0) = (0, \dots, 0)$. Similarly, $h' (0, \dots ,0) = (0, \dots, 0)$. By \ref{incl}, $k_i = l_i$ for all $i \in \{1, \dots,n\}$ and $h = h' = id$. This implies that $\alpha = id$ and $\xi = \zeta$.
\end{proof}

\begin{lem} \label{gweakmor}
Consider a weak morphism of precubical sets $f\colon |P| \to |P'|$, two morphisms of precubical sets $\chi \colon Q \to P$ and $\chi' \colon Q'\to P'$ and a continuous map $g\colon  |Q| \to |Q'|$ such that $f\circ |\chi| = |\chi'|\circ g$. Then $g$ is a weak morphism of precubical sets.
\end{lem}

\begin{proof}
Consider a vertex $v \in Q_0$. Let $x \in Q'$ and $u \in ]0,1[^ {\deg(x)}$ be the uniquely determined elements such that $g([v,()]) = [x,u]$. We have $[\chi'(x),u] = |\chi'|([x,u]) = |\chi'|\circ g ([v,()]) = f\circ |\chi|([v,()]) = f([\chi(v),()]) = [f_0(\chi(v)),()]$ and hence $\deg(x) = 0$ and $u=()$. 

Let $n, k_1,\dots, k_n \geq 1$ be integers and $\xi \colon \llbracket 0,{k_1} \rrbracket\otimes \cdots \otimes \llbracket 0,{k_n} \rrbracket \to Q$ be a morphism of precubical sets. Since $f$ is a weak morphism, there exist integers $k'_1,\dots, k'_n \geq 1$, a morphism of precubical sets $\theta \colon \llbracket 0,{k'_1} \rrbracket \otimes \cdots \otimes \llbracket 0,{k'_n} \rrbracket \to P'$ and a dihomeomorphism 
$\phi\colon  |\llbracket 0,{k_1} \rrbracket\otimes \cdots \otimes \llbracket 0,{k_n} \rrbracket|\to |\llbracket 0,{k'_1} \rrbracket\otimes \cdots \otimes \llbracket 0,{k'_n} \rrbracket|$ such that $f\circ |\chi\circ \xi| = |\theta|\circ \phi$. Consider the continuous map $\alpha = g \circ |\xi|\circ \phi^ {-1}\colon |\llbracket 0,{k'_1} \rrbracket\otimes \cdots \otimes \llbracket 0,{k'_n} \rrbracket| \to |Q'$. We have $\alpha \circ \phi = g\circ |\xi|$ and $|\chi'|\circ \alpha = |\chi'|\circ  g \circ |\xi|\circ \phi^ {-1} = f \circ |\chi| \circ |\xi|\circ \phi^ {-1} = |\theta|\circ \phi \circ \phi^{-1} = |\theta|$. By lemma \ref{reallem}, there exists a morphism of precubical sets $\xi'\colon \llbracket 0,{k'_1} \rrbracket\otimes \cdots \otimes \llbracket 0,{k'_n} \rrbracket \to Q'$ such that $\chi' \circ \xi' = \theta$ and $|\xi'| = \alpha$. The result follows.
\end{proof}

We are finally ready to establish the uniqueness of the maps $\chi'$ and $\phi$ in the definition of weak morphisms and the fact that $\phi$ is itself a weak morphism:

\begin{prop} \label{unicity}
Let $f\colon |P| \to |P'|$ be a weak morphism of precubical sets, $n, k_1, \dots, k_n\geq 1$ be integers and $\chi \colon \llbracket 0,{k_1} \rrbracket\otimes \cdots \otimes \llbracket 0,{k_n} \rrbracket \to P$ be  a morphism of precubical sets. Then there exist unique integers $k'_1, \dots, k'_n\geq 1$, a unique morphism of precubical sets $\chi' \colon \llbracket 0,{k'_1} \rrbracket\otimes \cdots \otimes \llbracket 0,{k'_n} \rrbracket \to P'$ and a unique dihomeo\-morphism $\phi\colon  |\llbracket 0,{k_1} \rrbracket\otimes \cdots \otimes \llbracket 0,{k_n} \rrbracket| \to |\llbracket 0,{k'_1} \rrbracket\otimes \cdots \otimes \llbracket 0,{k'_n} \rrbracket|$ such that $f\circ |\chi| = |\chi'|\circ \phi$. Moreover, $\phi$ is a weak morphism.
\end{prop}

\begin{proof}
The existence of $k'_1, \dots, k'_n$, $\chi' $ and $\phi$ is guaranteed by the definition of weak morphisms. The fact that $\phi$ is a weak morphism follows from \ref{gweakmor}. We have to show the uniqueness of $k'_1, \dots, k'_n$, $\chi' $ and $\phi$. Suppose that the integers $l_1, \dots, l_n \geq 1$, the morphism of precubical sets $\zeta  \colon \llbracket 0,{l_1} \rrbracket \otimes \cdots \otimes \llbracket 0,{l_n} \rrbracket \to Q$ and the dihomeo\-morphism $\psi \colon |\llbracket 0,{k_1} \rrbracket\otimes \cdots \otimes \llbracket 0,{k_n} \rrbracket|  \to |\llbracket 0,{l_1} \rrbracket \otimes \cdots \otimes \llbracket 0,{l_n} \rrbracket|$ satisfy $f\circ |\chi| = |\zeta |\circ \psi$. Consider the homeomorphism $\alpha = \psi \circ \phi^{-1} \colon |\llbracket 0,{k'_1} \rrbracket \otimes \cdots \otimes \llbracket 0,{k'_n} \rrbracket| \to |\llbracket 0,{l_1} \rrbracket \otimes \cdots \otimes \llbracket 0,{l_n} \rrbracket|$. Then $|\zeta| \circ \alpha = |\zeta|\circ \psi \circ \phi^{-1} = f\circ |\chi | \circ \phi^{-1} = |\chi'| \circ \phi \circ \phi^ {-1} = |\chi'|$. By lemma \ref{uniqueness}, it follows that $k'_i = l_i$ for all $i \in \{1, \dots,n\}$, $\phi = \psi$ and $\chi' = \zeta$.  
\end{proof}

\subsection{Subdivisions} \label{defsubdiv}
A \emph{subdivision} of a precubical set $P$ is a precubical set $Q$ together with a homeomorphism $f\colon |P| \to |Q|$ satisfying the following two conditions:
\begin{enumerate}
\item for every vertex $v\in P_0$ there exists a vertex $w\in Q_0$ such that $f([v,()]) = [w,()]$;
\item for every $n \geq 1$ and every element $x\in P_n$ there exist integers $k_1 \dots, k_n \geq 1$, a morphism of precubical sets $\xi \colon \llbracket 0,{k_1} \rrbracket \otimes \cdots \otimes \llbracket 0,{k_n} \rrbracket \to Q$ and dihomeomorphisms $\phi_i \colon |\llbracket 0,1 \rrbracket| = [0,1] \to |\llbracket 0,{k_i} \rrbracket| =  [0,k_i]$ $(i \in \{1, \dots, n\})$  such that $f\circ |x_{\sharp}| = |\xi| \circ (\phi_1 \times \cdots \times \phi_n)$.
\end{enumerate}

\begin{ex}
For all $n, k_1, \dots, k_n \geq 1$, the precubical set $\llbracket 0,{k_1} \rrbracket \otimes \cdots \otimes \llbracket 0,{k_n} \rrbracket$ together with the homeomorphism $$\begin{array}{rcl}f\colon |\llbracket 0,1 \rrbracket^ {\otimes n}| = [0,1]^ n &\to& |\llbracket 0,{k_1} \rrbracket \otimes \cdots \otimes \llbracket 0,{k_n} \rrbracket| =  [0,k_1] \times \cdots \times [0,k_n],\\ (t_1, \dots, t_n) &\mapsto& (k_1t_1, \dots, k_nt_n)\end{array}$$ is a subdivision of the precubical $n$-cube $\llbracket 0,1 \rrbracket^ {\otimes n}$. It is clear that condition (1) of the definition of subdivisions is satisfied. For condition  (2) consider an integer $m \geq 1$ and an element $x = (x_1, \dots, x_n) \in (\llbracket 0,1 \rrbracket^{\otimes n})_m$. Since $\deg(x) = m$, there exist indices $1 \leq i_1 < \ldots < i_m \leq n$ such that $x_{i_1} = \ldots = x_{i_m} = [0,1]$ and $x_i \in \{0,1\}$ for $i \notin \{i_1, \dots,i_m\}$. Consider the dihomeomorphisms $\phi_j\colon [0,1] \to [0,k_{i_j}]$, $t\mapsto k_{i_j}t$ $(j\in \{1, \dots, m\})$ and the morphism of precubical sets $\xi \colon \llbracket 0,{k_{i_1}} \rrbracket \otimes \cdots \otimes \llbracket 0,{k_{i_m}} \rrbracket \to \llbracket 0,{k_1} \rrbracket \otimes \cdots \otimes \llbracket 0,{k_n} \rrbracket$ defined by $\xi (a_1,\dots a_m) = (b_1, \dots, b_n)$ with $b_{i_j} = a_j$ for $j\in \{1, \dots, m\}$ and $b_i = k_ix_i$ for  $i \notin \{i_1, \dots,i_m\}$. Then $f\circ |x_{\sharp}| = |\xi|\circ (\phi_1 \times \cdots \times \phi_m)$.

\end{ex}

\begin{prop} \label{subweak}
Let $P$ be a precubical set and $(Q,f)$ be a subdivision of $P$. Then $f$ is a weak morphism.
\end{prop}

\begin{proof}
We only have to show that condition  (2) of the  definition of weak morphisms (see \ref{defweakmor})  holds. Let $n\geq 1$ be an integer. Consider integers $k_1, \dots, k_n \geq 1$ and a morphism of precubical sets $\chi \colon \llbracket 0,{k_1} \rrbracket\otimes \cdots \otimes \llbracket 0,{k_n} \rrbracket \to P$. We show that there exist integers $p_1, \dots, p_n \geq 1$, dihomeomorphisms $\phi_j\colon [0,k_j] \to [0,p_j]$ $(j \in \{1,\dots, n\})$ and  a morphism of precubical sets $\chi' \colon \llbracket 0,{p_1} \rrbracket \otimes \cdots \otimes \llbracket 0,{p_n} \rrbracket\to Q$ such that $f\circ |\chi| = |\chi'|\circ (\phi_1\times \cdots \times \phi_n)$. If all $k_i = 1$, then this holds by condition (2) of the definition of subdivisions. Suppose inductively that $\sum_{i=1}^nk_i > n$ and that the claim holds for all morphisms of precubical sets $\llbracket 0,{l_1} \rrbracket\otimes \cdots \otimes \llbracket 0,{l_n} \rrbracket \to P$ such that $\sum_{i=1}^nl_i < \sum_{i=1}^nk_i$. Let $s \in \{1, \dots, n\}$ be an index such that $k_s > 1$. The precubical set $\llbracket 0,{k_1} \rrbracket\otimes \cdots \otimes \llbracket 0,{k_n} \rrbracket$ is the union of the precubical subsets 
$$A = \llbracket 0,{k_1} \rrbracket\otimes \cdots \otimes \llbracket 0,{k_{s-1}} \rrbracket\otimes \llbracket 0,1 \rrbracket \otimes \llbracket 0,{k_{s+1}} \rrbracket\otimes \cdots \otimes \llbracket 0,{k_n} \rrbracket$$ and 
$$B = \llbracket 0,{k_1} \rrbracket\otimes \cdots \otimes \llbracket 0,{k_{s-1}} \rrbracket\otimes \llbracket 1,{k_s} \rrbracket  \otimes \llbracket 0,{k_{s+1}} \rrbracket\otimes \cdots \otimes \llbracket 0,{k_n}. \rrbracket$$
Denote the restrictions of $\chi$ to $A$ and $B$ by $\alpha$ and $\beta$ respectively. By the inductive hypothesis, there exist integers $q_1, \dots, q_n \geq 1$, dihomeomorphisms $\sigma_s\colon [0,1] \to [0,q_s]$ and $\sigma_j\colon [0,k_j] \to [0,q_j]$ $(j\in \{1,\dots,n\}\setminus\{s\})$  and a morphism of precubical sets $\alpha' \colon \llbracket 0,{q_1} \rrbracket\otimes \cdots \otimes \llbracket 0,{q_n} \rrbracket \to Q$ such that $f\circ |\alpha| = |\alpha'|\circ (\sigma_1 \times \cdots \times \sigma_n)$. 

Consider the isomorphism of precubical sets 
$$\eta\colon \llbracket 0,{k_1} \rrbracket\otimes \cdots \otimes \llbracket 0,{k_{s-1}} \rrbracket\otimes  \llbracket 0,{k_s-1} \rrbracket \otimes \llbracket 0,{k_{s+1}} \rrbracket\otimes \cdots \otimes \llbracket 0,{k_n} \rrbracket \to B$$
given by $$(y_1, \dots, y_{s-1},i,y_{s+1}, \dots, y_n) \mapsto (y_1, \dots, y_{s-1},i+1,y_{s+1}, \dots, y_n)$$ for $i\in \{0, \dots, k_s-1\}$ and by $$(y_1, \dots, y_{s-1},[i-1,i],y_{s+1}, \dots, y_n) \mapsto (y_1, \dots, y_{s-1},[i,i+1],y_{s+1}, \dots, y_n)$$ for $i\in \{1, \dots, k_s-1\}$.  By the inductive hypothesis, there exist integers $r_1, \dots, r_n \geq 1$, dihomeomorphisms $\psi_j\colon [0,k_j] \to [0,r_j]$ $(j\in \{1,\dots,m\}\setminus\{s\})$ and $\psi_s\colon [0,k_s-1] \to [0,r_s]$ and a morphism of precubical sets $\gamma \colon \llbracket 0,{r_1} \rrbracket\otimes \cdots \otimes  \llbracket 0,{r_n} \rrbracket \to Q$ such that $f\circ |\beta \circ \eta| = |\gamma|\circ (\psi_1 \times \cdots \times \psi_n)$. 

Consider the precubical set $$B' = \llbracket 0,{r_1} \rrbracket\otimes \cdots \otimes \llbracket 0,{r_{s-1}} \rrbracket\otimes  \llbracket q_s,{q_s+r_s} \rrbracket  \otimes \llbracket 0,{r_{s+1}} \rrbracket\otimes \cdots \otimes \llbracket 0,{r_n} \rrbracket$$ and the isomorphism of precubical sets $\mu\colon \llbracket 0,{r_1} \rrbracket\otimes \cdots \otimes \llbracket 0,{r_n} \rrbracket \to B'$ given by $$(y_1, \dots, y_{s-1},i,y_{s+1}, \dots, y_n) \mapsto (y_1, \dots, y_{s-1},q_s+i,y_{s+1}, \dots, y_n)$$ for $i\in \{0, \dots, r_s\}$ and by $$(y_1, \dots, y_{s-1},[i-1,i],y_{s+1}, \dots, y_n) \mapsto (y_1, \dots, y_{s-1},[q_s+i-1,q_s+i],y_{s+1}, \dots, y_n)$$ for $i\in \{1, \dots, r_s\}$. Consider the morphism of precubical sets $\beta' \colon B' \to Q$ defined by $\beta ' = \gamma \circ \mu^ {-1}$ and the dihomeomorphism $\theta_s\colon [1, k_s] \to [q_s,q_s+r_s]$ given by $\theta_s(t) = q_s+\psi_s(t-1)$. For $j \in \{1, \dots ,n\}\setminus \{s\}$ set $\theta_j = \psi_j$. Then $f\circ |\beta | = |\beta'|\circ (\theta_1 \times \cdots \times \theta_n)$.

Set $A' = \llbracket 0,{q_1} \rrbracket\otimes \cdots \otimes \llbracket 0,{q_n} \rrbracket$. We show that $q_j = r_j$ and $\sigma _j = \theta_j$ for all $j \in \{1, \dots, n\} \setminus \{s\}$ and that 
$\alpha'$ and $\beta'$ coincide on $A'\cap B'$. If $n= 1$, we only have to show that $\alpha'$ and $\beta'$ coincide on $A'\cap B' = \{q_s\}$. We have $|\alpha'|(q_s) = |\alpha'|\circ \sigma_s(1) = f\circ |\alpha| (1) = f \circ |\chi| (1) = f\circ |\beta|(1) = |\beta '| \circ \theta_s(1) = |\beta'|(q_s)$. Thus, $\alpha'(q_s) = \beta'(q_s)$. Suppose now that $n > 1$. Consider the composite morphisms of precubical sets
$$\xi\colon \llbracket 0,{k_1} \rrbracket\otimes \cdots \otimes \llbracket 0,{k_{s-1}} \rrbracket \otimes \llbracket 0,{k_{s+1}} \rrbracket\otimes \cdots \otimes \llbracket 0,{k_n} \rrbracket \iso A\cap B \hookrightarrow A$$
and \begin{eqnarray*}\lefteqn{\xi'\colon \llbracket 0,{q_1} \rrbracket\otimes \cdots \otimes \llbracket 0,{q_{s-1}} \rrbracket \otimes \llbracket 0,{q_{s+1}} \rrbracket\otimes \cdots \otimes \llbracket 0,{q_n} \rrbracket }\\ &\iso & \llbracket 0,{q_1} \rrbracket\otimes \cdots \otimes \llbracket 0,{q_{s-1}} \rrbracket \otimes \{q_s\}\otimes \llbracket 0,{q_{s+1}} \rrbracket\otimes \cdots \otimes \llbracket 0,{q_n} \rrbracket\\ &\hookrightarrow& A'.\end{eqnarray*} We have $ (\sigma_1 \times \cdots \times \sigma_n)\circ |\xi|=|\xi'|\circ (\sigma_1 \times \cdots \times \sigma_{s-1}\times \sigma_{s+1}\times \cdots \times  \sigma_n) $ and hence $f\circ |\alpha \circ \xi| = |\alpha' \circ \xi'|\circ (\sigma_1 \times \cdots \times \sigma_{s-1}\times \sigma_{s+1}\times \cdots \times \sigma_n)$. Consider the composite morphisms of precubical sets
$$\zeta\colon \llbracket 0,{k_1} \rrbracket\otimes \cdots \otimes \llbracket 0,{k_{s-1}} \rrbracket \otimes \llbracket 0,{k_{s+1}} \rrbracket\otimes \cdots \otimes \llbracket 0,{k_n} \rrbracket \iso A\cap B \hookrightarrow B$$
and \begin{eqnarray*}\lefteqn{\zeta'\colon \llbracket 0,{r_1} \rrbracket\otimes \cdots \otimes \llbracket 0,{r_{s-1}} \rrbracket \otimes \llbracket 0,{r_{s+1}} \rrbracket\otimes \cdots \otimes \llbracket 0,{r_n} \rrbracket }\\ &\iso & \llbracket 0,{r_1} \rrbracket\otimes \cdots \otimes \llbracket 0,{r_{s-1}} \rrbracket \otimes \{q_s\}\otimes \llbracket 0,{r_{s+1}} \rrbracket\otimes \cdots \otimes \llbracket 0,{r_n} \rrbracket\\ &\hookrightarrow& B'.\end{eqnarray*} We have $(\theta_1 \times \cdots \times \theta_n)\circ |\zeta| = |\zeta'|\circ (\theta_1 \times \cdots \times \theta_{s-1}\times \theta_{s+1}\times \cdots \times  \theta_n)$ and hence $f\circ |\beta \circ \zeta| = |\beta' \circ \zeta'|\circ (\theta_1 \times \cdots \times \theta_{s-1}\times \theta_{s+1}\times \cdots \times \theta_n)$. Since $\alpha \circ \xi = \beta \circ \zeta$, we have $|\beta'\circ \zeta'| \circ (\theta_1 \times \cdots \times \theta_{s-1}\times \theta_{s+1}\times \cdots \times \theta_n)\circ  (\sigma_1 \times \cdots \times \sigma_{s-1}\times \sigma_{s+1}\times \cdots \times  \sigma_n)^ {-1} = |\alpha' \circ \xi'|$.  It follows by \ref{uniqueness} that for all $j \in \{1, \dots, n\} \setminus \{s\}$, $q_j = r_j$ and $\sigma _j = \theta_j$. Moreover, $\alpha' \circ \xi' = \beta' \circ \zeta'$. Since $\alpha' \circ \xi' = \beta' \circ \zeta'$, $\alpha'$ and $\beta'$ coincide on $A'\cap B'$. 

Set $p_s=q_s+r_s$ and $p_j = q_j=r_j$ $(j \in \{1, \dots, n\} \setminus \{s\})$. Then we have $\llbracket 0,{p_1} \rrbracket\otimes \cdots \otimes \llbracket 0,{p_n} \rrbracket = A' \cup B'$. Let $\chi'\colon \llbracket 0,{p_1} \rrbracket\otimes \cdots \otimes \llbracket 0,{p_n} \rrbracket \to Q$ be the unique morphism of precubical sets such that $\chi'|_{A'} = \alpha'$ and $\chi'|_{B'} = \beta'$. Consider the dihomeomorphism $\phi_s\colon [0,k_s] \to [0,p_s]$ defined by
$$\phi_s(t) = \left\{\begin{array}{ll} \sigma_s (t),& 0 \leq t \leq 1,\\ \theta_s(t), & 1\leq t \leq k_s. \end{array}\right.$$
For $j \in \{1, \dots, n\} \setminus \{s\}$ set $\phi_j = \sigma_j = \theta_j$. We have $f\circ |\chi| = |\chi'|\circ (\phi_1 \times \cdots \times \phi_n)$. This terminates the induction and the proof.
\end{proof}

\subsection{Weak morphisms and paths} Let $f \colon |P| \to |P'|$ be a weak morphism of precubical sets and $\omega \colon \llbracket 0,k \rrbracket \to P$ $(k \geq 0)$ be a path. If $k > 0$, we denote by $f^{\mathbb I}(\omega)$ the unique path $\omega ' \colon \llbracket 0,{k'} \rrbracket \to P'$ for which there exists a dihomeomorphism $\phi \colon |\llbracket 0,k \rrbracket| \to |\llbracket 0,{k'} \rrbracket|$ such that $f\circ |\omega | = |\omega' |\circ \phi$. If $k= 0$, $f^{\mathbb I}(\omega)$ is defined to be the path in $P'$ of length $0$ given by $f^{\mathbb I}(\omega)(0) = f_0(\omega(0))$. Note that if $g\colon P \to P'$ is a morphism of precubical sets, then $|g|^{\mathbb I}(\omega) = g\circ \omega$. The next three propositions contain the basic properties of the maps $f^{\mathbb I}\colon P^{\mathbb I} \to P'^{\mathbb I}$. 

\begin{prop} \label{fomega} Let $f \colon |P| \to |P'|$ be a weak morphism of precubical sets and $\omega \colon \llbracket 0,k \rrbracket \to P$ be a path. Then  $f^{\mathbb I}(\omega )$ is a path from $f_0(\omega(0))$ to $f_0(\omega(k))$.
\end{prop}

\begin{proof}
This is clear for $k=0$. Suppose that $k > 0$ and set $k' = \length(f^{\mathbb I}(\omega))$. Let $\phi \colon |\llbracket 0,k \rrbracket| = [0,k] \to |\llbracket 0,{k'} \rrbracket| = [0,k']$ be the dihomeomorphism such that $f\circ |\omega | = |f^{\mathbb I}(\omega) |\circ \phi$. Since $\phi$ is a dihomeomorphism, $\phi(0) = 0$ and $\phi(k) = k'$. We therefore have $[f^{\mathbb I}(\omega)(0),()] = |f^{\mathbb I}(\omega)|(0) = |f^{\mathbb I}(\omega)|\circ \phi(0) = f\circ |\omega|(0) = f([\omega (0),()]) = [f_0(\omega (0)),()]$ and $[f^{\mathbb I}(\omega)(k'),()] = |f^{\mathbb I}(\omega)|(k') = |f^{\mathbb I}(\omega)|\circ \phi(k) = f\circ |\omega|(k) = f([\omega (k),()]) = [f_0(\omega (k)),()]$. It follows that $f^{\mathbb I}(\omega)(0) = f_0(\omega(0))$ and $f^{\mathbb I}(\omega)(k') = f_0(\omega(k))$. 
\end{proof}

\begin{prop} \label{fJfunctor}
Let $f\colon |P| \to |Q|$ and $g\colon |Q| \to |R|$ be weak morphisms of precubical sets. Then $(g\circ f)^{\mathbb I} = g^{\mathbb I} \circ f^{\mathbb I}$.
\end{prop}

\begin{proof}
Let $\omega \colon \llbracket 0,k \rrbracket \to P$ be a path. Suppose first that $k=0$. Then $f^{\mathbb I}(\omega)$, $g^{\mathbb I}\circ f^{\mathbb I}(\omega)$ and $(g\circ f)^{\mathbb I}(\omega)$ are paths of length $0$. We have $$g^{\mathbb I}\circ f^{\mathbb I}(\omega) (0) = g_0\circ f_0 (\omega (0)) = (g\circ f)_0(\omega (0)) = (g\circ f)^{\mathbb I}(\omega)(0)$$ and therefore $(g\circ f)^{\mathbb I}(\omega)  = g^{\mathbb I} \circ f^{\mathbb I}(\omega)$. Suppose that $k > 0$ and set $k' = \length(f^{\mathbb I}(\omega))$ and $k'' = \length(g^{\mathbb I}(f^{\mathbb I}(\omega)))$. Let $\phi \colon |\llbracket 0,k \rrbracket| \to |\llbracket 0,{k'} \rrbracket|$ and $\psi \colon |\llbracket 0,{k'} \rrbracket| \to |\llbracket 0,{k''} \rrbracket|$ be dihomeomorphisms such that $f\circ |\omega | = |f^{\mathbb I}(\omega) |\circ \phi$ and $g\circ |f^{\mathbb I}(\omega)| = |g^{\mathbb I}(f^{\mathbb I}(\omega)) |\circ \psi$. Then $g \circ f\circ |\omega | = |g^{\mathbb I}(f^{\mathbb I}(\omega)) |\circ \psi \circ \phi$. Since $\psi \circ \phi$ is a dihomeomorphism, this implies that $(g\circ f)^{\mathbb I}(\omega) = g^{\mathbb I}(f^{\mathbb I}(\omega))$.
\end{proof}

\begin{prop} \label{fJconcat}
Let $f\colon |P| \to |Q|$ be a weak morphism of precubical sets and $\omega\colon \llbracket 0,k \rrbracket \to P$ and $\nu \colon \llbracket 0,l \rrbracket \to P$ be paths in $P$ such that $\omega (k) = \nu (0)$. Then $f^{\mathbb I}(\omega \cdot \nu) = f^{\mathbb I}(\omega ) \cdot f^{\mathbb I}(\nu)$.
\end{prop}

\begin{proof}
We may suppose that $k,l >0$. Suppose that $f^{\mathbb I}(\omega)$ and $f^{\mathbb I}(\nu)$ are of length $k'$ and $l'$ respectively. Note that by \ref{fomega}, $f^{\mathbb I}(\omega) (k') = f_0(\omega(k)) = f_0(\nu(0)) = f^{\mathbb I}(\nu)(0)$. Let $\phi \colon |\llbracket 0,k \rrbracket| \to |\llbracket 0,{k'} \rrbracket|$ and $\psi \colon |\llbracket 0,l \rrbracket| \to |\llbracket 0,{l'} \rrbracket|$ be dihomeomorphisms such that $|f^{\mathbb I}(\omega)| \circ \phi = f \circ |\omega |$ and $|f^{\mathbb I}(\nu)| \circ \psi = f \circ |\nu|$. Consider the dihomeomorphism $\alpha \colon |\llbracket 0,{k+l} \rrbracket| = [0,k+l] \to |\llbracket 0,{k'+l'} \rrbracket|= [0,k'+l']$ defined by $$\alpha (t) = \left\{ \begin{array}{ll} \phi (t), & 0 \leq t \leq k,\\ k'+\psi (t-k), & k \leq t \leq k+l. \end{array} \right.$$ For $0 \leq t \leq k$, $|f^{\mathbb I}(\omega)\cdot f^{\mathbb I}(\nu)|\circ \alpha (t) = |f^{\mathbb I}(\omega)\cdot f^{\mathbb I}(\nu)| (\phi (t)) = |f^{\mathbb I}(\omega)|(\phi(t)) =$ \linebreak  $f\circ |\omega| (t) = f \circ |\omega \cdot \nu |(t)$. For $k \leq t \leq k+l$, $|f^{\mathbb I}(\omega)\cdot f^{\mathbb I}(\nu)|\circ \alpha (t) = $\linebreak $|f^{\mathbb I}(\omega)\cdot f^{\mathbb I}(\nu)| (k'+\psi (t-k)) = |f^{\mathbb I}(\nu)| (\psi (t-k)) = f \circ |\nu| (t-k) = f \circ |\omega \cdot \nu| (t)$. We therefore have $|f^{\mathbb I}(\omega)\cdot f^{\mathbb I}(\nu)|\circ \alpha = f \circ |\omega \cdot \nu|$ and hence $f^{\mathbb I}(\omega \cdot \nu) = f^{\mathbb I}(\omega ) \cdot f^{\mathbb I}(\nu)$.
\end{proof}

\subsection{Weak morphisms of HDAs }
 \label{defHDAmor}
A \emph{weak morphism} from an $M$-HDA  $\A = (P,I, F,\lambda)$ to an $M$-HDA $\B = (P',I', F',\lambda')$  is a weak morphism $f\colon |P| \to |P'|$  such that $f_0(I) \subseteq I'$, $f_0(F) \subseteq F'$ and  $\overline{\lambda'}\circ f^{\mathbb I} = \overline{\lambda}$. We say that an $M$-HDA $\B = (P',I', F',\lambda')$ is a \emph{subdivision} of an $M$-HDA  $\A = (P,I, F,\lambda)$ if there exists a weak morphism $f$ from $\A$ to $\B$ such that $f_0(I) = I'$, $f_0(F) = F'$ and $(P',f)$ is a subdivision of $P$.

We remark that a morphism from an $M$-HDA $\A = (P,I, F,\lambda)$ to an $M$-HDA $\B = (P',I', F',\lambda')$ is a weak morphism. More preceisely, if $g\colon P \to P'$ is a morphism of precubical sets such that $g(I) \subseteq I'$, $g(F) \subseteq F'$ and $\lambda' \circ g = \lambda$, then $|g|$ is a weak morphism of $M$-HDAs.

\begin{defin} 
Let $\A$ and $\B$ be two $M$-HDAs. We write $A \to B$ if there exists a weak morphism from $\A$ to $\B$. We write $\A \leftrightarrow \B$ if $\A \to \B$ and $\B \to \A$.
\end{defin}

It follows from \ref{composition} and \ref{fJfunctor} that $\to$ and $\leftrightarrow$ are, respectively, a preorder and an equivalence relation on the class of $M$-HDAs.  

\begin{ex}
Consider the $1$-dimensional precubical set $P$ given by $P_0 = \{u,v,w\}$, $P_1 = \{x,y,z\}$, $d_1^0x = d_1^0z = u$, $d_1^1x = d_1^0y = v$ and $d_1^1y = d_1^1z = w$. Consider the set of labels $\Sigma = \{a,b\}$ and the $\Sigma^*$-HDAs $\A = (P,\{u\},\{w\},\lambda)$ and $\B = (P\setminus \{z\},\{u\},\{w\},\lambda|_{\{x,y\}})$ where the labelling function $\lambda$ is given by $\lambda (x) = a$, $\lambda (y) = b$ and $\lambda (z) = ab$. Then $\A \leftrightarrow \B$. Note that there is no morphism of $\Sigma^*$-HDAs from $\A$ to $\B$.
\end{ex}

\begin{prop} \label{lalb}
$\A \to \B \Rightarrow L(\A) \subseteq L(\B)$ and $\A \leftrightarrow \B \Rightarrow L(\A) = L(\B)$. 
\end{prop}

\begin{proof}
Let $\A = (P,I, F,\lambda)$ and $\B = (P',I', F',\lambda')$ be $M$-HDAs such that $\A \to \nolinebreak \B$. Let $f\colon |P| \to |P'|$ be a weak morphism from $\A$ to $\B$. Consider a path $\omega \colon \llbracket 0,k\rrbracket \to P$ such that $\omega (0) \in I$ and $\omega (k) \in F$. Then $f^{\mathbb I}(\omega )(0) = f_0(\omega(0)) \in I'$ and $f^{\mathbb I}(\omega )(\length(f^{\mathbb I}(\omega) )) = f_0(\omega (k)) \in F'$. We have $\overline{\lambda} (\omega) = \overline{\lambda'}(f^{\mathbb I}(\omega)) \in L(\B)$. This shows that $L(\A) \subseteq L(\B)$. The result follows.
\end{proof}

\section{Traces} \label{HDA}

We define the trace language and the trace category of an HDA and study how these objects behave under weak morphisms. The fundamental concept for both constructions is dihomotopy of paths. We also introduce trace equivalences of HDAs and the preorder relation of trace equivalent abstraction. The trace category is invariant under trace equivalent abstraction by definition, and we show that this is often also the case for the trace language.

\subsection{Dihomotopy} Two paths $\omega$ and $\nu $ in a precubical set $P$ are said to be \emph{elementarily dihomotopic} (or \emph{contiguous} or \emph{adjacent}) if there exist paths $\alpha, \beta \in P^ {\mathbb I}$ and an element $z \in P_2$ such that $d_1^0d_1^0z = \alpha (\length(\alpha))$, $d_1^1d_1^1z = \beta (0)$ and $\{\omega ,\nu \} = \{\alpha \cdot (d_1^0z)_{\sharp} \cdot (d_2^ 1z)_{\sharp} \cdot \beta, \alpha \cdot (d_2^0z)_{\sharp} \cdot (d_1^ 1z)_{\sharp} \cdot \beta \}$ (see   \cite{FajstrupGR, Goubault}). The \emph{dihomotopy} relation, denoted by $\sim$, is the equivalence relation generated by elementary dihomotopy.

\begin{rems}
(i) For a notion of homotopy for more general paths (essentially the cube paths we consider in section \ref{homeoabs}) the reader is referred to \cite{vanGlabbeek}.

(ii) Dihomotopic paths have the same end points and the  same length.

(iii) Dihomotopy is a congruence. More precisely, if $\alpha, \beta, \omega$ and $\nu$ are paths in a precubical set such that $\beta (0) = \alpha (\length(\alpha))$, $\nu (0)= \omega(\length(\omega))$, $\alpha \sim \omega$ and $\beta \sim \nu$, then $\alpha \cdot \beta \sim \omega \cdot \nu$.

(iv) If $\alpha $ and $\omega $ are (elementarily) dihomotopic paths in a precubical set $P$ and $f\colon P \to Q$ is a morphism of precubical sets, then $f\circ \alpha$ and $f \circ \omega$ are (elementarily) dihomotopic paths in $Q$.
\end{rems}

\subsection{The trace language of an HDA} \label{indep}  
Let $\A = (P,I, F,\lambda)$ be an $M$-HDA. We say that $a,b \in M$ are \emph{independent in $\A$} if 
\begin{enumerate}
\item $a,b \in \lambda(P_1)$;
\item $a\not= b$;
\item 
for all paths $\omega \in P^{\mathbb I}$ and elements $u,v \in M$ with $\overline{\lambda} (\omega ) \in \{uabv, ubav\}$ there exists a path $\nu \in P^{\mathbb I}$ such that $\omega$ and $\nu$ are dihomotopic and $\{\overline{\lambda} (\omega ),\overline{\lambda}(\nu)\}  = \{uabv, ubav\}$.
\end{enumerate}
If two elements of $M$ are not independent in $\A$, we say that they are \emph{dependent in $\A$}. We denote by $\equiv_{\A}$ the smallest congruence relation in $M$ such that $ab$ and $ba$ are congruent for all independent elements $a,b\in M$. The quotient monoid $M/\equiv_{\A}$ is called the \emph{trace monoid} of $\A$, and the 
canonical projection $M \to M/\equiv_{\A}$ is denoted by $tr_{\A}$. The \emph{trace language} of $\A$ is the set $TL(\A) = tr_{\A}(L(\A)) \subseteq M/\equiv_{\A}$.

\begin{rems}  
(i) We have $x \equiv_{\A} y$ if and only if there exist elements $x_1, \dots, x_n \in M$ $(n \geq 1)$ such that $x= x_1$, $y = x_n$ and for each $i \in \{1, \dots , n-1\}$, $x_i = u_ia_ib_iv_i$ and $x_{i+1} = u_ib_ia_iv_i$ for certain elements $a_i,b_i,u_i,v_i \in M$ with $a_i$ and $b_i$ independent.

(ii) If $M = \Sigma ^*$ for some finite set $\Sigma$ and $\lambda (P_1) \subseteq \Sigma$, then the set $$D = \{(a,b) \in \Sigma \times \Sigma : a\; \mbox{and}\;  b\; \mbox{are dependent in}\; \A\}$$ is a dependency in $\Sigma$ and $M/\equiv_{\A}$ is the classical trace monoid $M(\Sigma,D)$ (see e.g. \cite{Diekert, Mazurkiewicz}). Moreover, in this situation, condition (3) of the definition of independence can be replaced by the following condition: 
\begin{itemize}
\item[(3')] for any path $\omega =x_{1\sharp}x_{2\sharp} \colon \llbracket 0,2 \rrbracket \to P$ with $\{\lambda(x_1),\lambda(x_2)\} = \{a,b\}$ there exists an element $z\in P_2$ such that either $d_1^0z = x_1$ and $d_2^1z = x_2$ or $d_2^0z = x_1$ and $d_1^1z = x_2$.
\end{itemize}
\end{rems}

\begin{lem} \label{labellem}
Let $\A = (P,I, F,\lambda)$ be an $M$-HDA, $\omega \in P^{\mathbb I}$ and $m \in M$ such that $\overline{\lambda} (\omega) \equiv_{\A} m$. Then there exists a path $\nu \in P^{\mathbb I}$ such that $\omega \sim \nu$ and $\overline{\lambda} (\nu) = m$.
\end{lem}

\begin{proof}
Since $\overline{\lambda} (\omega) \equiv_{\A} m$, there exist $x_1, \dots, x_n \in M$ such that $\overline{\lambda} (\omega) = x_1$, $m = x_n$ and for all $i \in \{1, \dots , n-1\}$, $x_i = u_ia_ib_iv_i$ and $x_{i+1} = u_ib_ia_iv_i$ for certain elements $a_i,b_i,u_i,v_i \in M$ with $a_i$ and $b_i$  independent. A simple induction using condition (3) of the definition of independence establishes the result.
\end{proof}  

\begin{lem} \label{LATLA}
Let $\A = (P,I, F,\lambda)$ be an $M$-HDA. Then $L(\A)  = tr_{\A}^{-1}(TL(\A))$.
\end{lem}

\begin{proof}
It is enough to show that $tr_{\A}^{-1}(TL(\A)) \subseteq L(\A)$. Consider an element $m \in tr_{\A}^{-1}(TL(\A))$. Then there exists a path $\omega \in P^{\mathbb I}$ beginning in $I$ and ending in $F$ such that $tr_{\A}(\overline{\lambda}(\omega)) = tr_{\A}(m)$, i.e. $\overline{\lambda} (\omega) \equiv_{\A} m$. By the preceding lemma, there exists a path $\nu \in P^{\mathbb I}$ with the same end points as $\omega$ and  $\overline{\lambda} (\nu) = m$. Thus, $m \in L(\A)$.
\end{proof}  

\begin{theor}
If $\A = (P,I,F, \lambda)$ is an $M$-HDA such that $P_0$ and $P_1$ are finite, then $TL(\A)$ is a rational subset of the monoid $M/\equiv_{\A}$. If, moreover, $M$ is a finitely generated free monoid, then $TL(\A)$ is a recognisable subset of $M/\equiv_{\A}$.
\end{theor}

\begin{proof}
It is well-known that the image of a rational set under a morphism of monoids is rational and that a subset of a monoid is recognisable if its preimage under a surjective morphism of monoids is recognisable \cite[prop. II.1.7, cor. II.2.12]{Sakarovitch}. The result now follows from \ref{LA} and \ref{LATLA}.
\end{proof}

\subsection{Stable HDAs}
In a general $M$-HDA it is possible that dependent actions are locally independent. In stable $M$-HDAs, which we define next, this is impossible.

\begin{defin}
We say that an $M$-HDA $\A = (P,I,F, \lambda)$ is \emph{stable} if for no two dependent elements $a,b \in M$ there exists an element $z \in P_2$ such that $\{\lambda (d_1^0z),\lambda (d_2^0z)\} = \{a,b\}$.
\end{defin}

\begin{prop} \label{stablelabel}
Let $\A = (P,I,F, \lambda)$ be a stable $M$-HDA. Then any two dihomotopic paths have congruent labels.
\end{prop}

\begin{proof} 
Consider elementarily dihomotopic paths $\omega = \alpha \cdot (d_1^0z)_{\sharp} \cdot (d_2^ 1z)_{\sharp} \cdot \beta$ and $\nu = \alpha \cdot (d_2^0z)_{\sharp} \cdot (d_1^ 1z)_{\sharp} \cdot \beta$. Since $\A$ is stable, the elements $\lambda(d_1^0z) = \lambda (d_1^1z)$ and  $\lambda(d_2^0z) = \lambda (d_2^1z)$ are independent. Therefore $\lambda(d_1^0z)\lambda(d_2^1z) \equiv_{\A} \lambda(d_2^0z)\lambda(d_1^1z)$. Hence $\overline{\lambda}(\omega ) = \overline{\lambda}(\alpha ) \lambda(d_1^0z)\lambda(d_2^1z) \overline{\lambda}(\beta )\equiv_{\A} \overline{\lambda}(\alpha ) \lambda(d_2^0z)\lambda(d_1^1z)\overline{\lambda}(\beta ) = \overline{\lambda}(\nu )$.
\end{proof}

\subsection{Deterministic HDAs} 

A converse of \ref{stablelabel} holds in deterministic HDAs, which are defined as follows:

\begin{defin}  
We say that an $M$-HDA is \emph{deterministic} if it has exactly one initial state and if any two paths with the same starting point and the same label are equal.
\end{defin}

\begin{prop} \label{weakmordet}
Let $\A = (P,I,F, \lambda)$ be a deterministic $M$-HDA. Then any two paths beginning at the same vertex and with  congruent labels are dihomotopic.
\end{prop}

\begin{proof}
Let $\omega, \nu \in P^{\mathbb I}$ be two paths such that $\omega (0) = \nu (0)$ and $\overline{\lambda}(\omega) \equiv_{\A} \overline{\lambda}(\nu)$. By \ref{labellem}, there exists a path $\alpha \in P^{\mathbb I}$ such that $\omega \sim \alpha$ and $\overline{\lambda} (\alpha) = \overline{\lambda}(\nu)$. Since $\alpha(0) = \omega(0) = \nu (0)$, $\alpha = \nu$.
\end{proof}

\begin{cor}
Let $\A = (P,I,F, \lambda)$ be a deterministic and stable $M$-HDA. Then two paths with the same starting point are dihomotopic if and only if they have congruent labels.
\end{cor}

\subsection{Subdivided cubes} \label{detcube}
 
Any subdivided cube $\llbracket 0,{k_1} \rrbracket \otimes \cdots \otimes \llbracket 0,{k_n} \rrbracket$ $(n,k_1, \dots , k_n \geq 1)$ is the underlying precubical set of a stable and deterministic HDA $\A$ over the free monoid on the set $\Sigma = \{l_1,\dots, l_n\}$. The labelling function of $\A$ is given by $$\lambda(j_1, \dots ,j_{p-1},[j_p,j_p+1],j_{p+1}, \dots, j_n) = l_p,$$ 
and the sets of initial and final states are $I = \{(0, \dots, 0)\}$ and $F = \{(k_1, \dots, k_n)\}$. Any two distinct elements of $\Sigma$ are independent in $\A$, and for any element $z$ of degree $2$, $\lambda (d_1^0z)\not= \lambda (d_2^0z)$. Consequently, $\A$ is stable. It follows from the fact that any vertex of $\llbracket 0,{k_1} \rrbracket \otimes \cdots \otimes \llbracket 0,{k_n} \rrbracket$ is the starting point of at most one edge with a given label that $\A$ is deterministic. One easily sees that given a path $\omega \colon \llbracket 0,r \rrbracket \to \llbracket 0,{k_1} \rrbracket \otimes \cdots \otimes \llbracket 0,{k_n} \rrbracket$ from  $(a_1, \dots ,a_n)$ to $(b_1, \dots , b_n)$, one has $b_i \geq  a_i$ and $\overline{\lambda }(\omega) \equiv_{\A} l_1^{b_1-a_1}\cdots l_n^{b_n-a_n}$. Consequently, any two paths in $\A$ with the same end points have congruent labels and are dihomotopic.

\subsection{Trace languages and weak morphisms}

The existence of a weak morphism between two $M$-HDAs permits us to compare their trace languages. The main point is that weak morphisms preserve dihomotopy:

\begin{prop} \label{fadj}
Let $f\colon P \to Q$ be a weak morphism of precubical sets and $\omega$ and $\nu$ be dihomotopic paths in $P$. Then $f^{\mathbb I}(\omega)$ and $f^{\mathbb I}(\nu)$ are dihomotopic paths in $Q$.
\end{prop}

\begin{proof}
We may suppose that $\omega$ and $\nu$ are elementarily dihomotopic and that there exist an element $z \in P_2$ and paths $\alpha$ and $\beta$  such that $\omega = \alpha \cdot (d_1^0z)_{\sharp }\cdot (d_2^1z)_{\sharp }\cdot \beta$ and $\nu = \alpha \cdot (d_2^0z)_{\sharp }\cdot (d_1^1z)_{\sharp }\cdot \beta$. By \ref{fJconcat}, $f^{\mathbb I}(\omega) = f^{\mathbb I}(\alpha )\cdot f^{\mathbb I}((d_1^0z)_{\sharp }\cdot (d_2^1z)_{\sharp })\cdot f^{\mathbb I}(\beta)$ and $f^{\mathbb I}(\nu) = f^{\mathbb I}(\alpha )\cdot f^{\mathbb I}((d_2^0z)_{\sharp }\cdot (d_1^1z)_{\sharp })\cdot f^{\mathbb I}(\beta)$. It follows that $f^{\mathbb I}(\omega) \sim f^{\mathbb I}(\nu)$ if $f^{\mathbb I}((d_1^0z)_{\sharp }\cdot (d_2^1z)_{\sharp }) \sim f^{\mathbb I}((d_2^0z)_{\sharp }\cdot (d_1^1z)_{\sharp })$. We may thus suppose that  $\omega = (d_1^0z)_{\sharp }\cdot (d_2^1z)_{\sharp }$ and $\nu = (d_2^0z)_{\sharp }\cdot (d_1^1z)_{\sharp }$. Consider the morphism of precubical sets $z_{\sharp} \colon \llbracket 0,1 \rrbracket^{\otimes 2}  \to P$. Since $f$ is a weak morphism, there exist integers $k,l \geq 1$, a morphism of precubical sets $\zeta \colon \llbracket 0,k \rrbracket\otimes \llbracket 0,l \rrbracket \to Q$  and a dihomeomorphism $$\phi \colon |\llbracket 0,1 \rrbracket^{\otimes 2}| = [0,1]^2 \to |\llbracket 0,k \rrbracket\otimes \llbracket 0,l \rrbracket| = [0,k]\times [0,l]$$ such that $f\circ |z_{\sharp}| = |\zeta| \circ \phi$. Consider the paths $\sigma = (d_1^0\iota_2)_{\sharp }\cdot (d_2^1\iota_2)_{\sharp }$ and $\tau = (d_2^0\iota_2)_{\sharp }\cdot (d_1^1\iota_2)_{\sharp }$ in $\llbracket 0,1 \rrbracket^{\otimes 2}$. By \ref{fomega},  $\phi^{\mathbb I}(\sigma)$ and $\phi^{\mathbb I}(\tau)$ are paths in $\llbracket 0,k \rrbracket\otimes \llbracket 0,l \rrbracket$ from $\phi_0(0,0) = (0,0)$ to $\phi_0(1,1) = (k,l)$. Thus, $\phi^{\mathbb I}(\sigma) \sim \phi^{\mathbb I}(\tau)$ (see \ref{detcube}). It follows that $\zeta \circ \phi^{\mathbb I}(\sigma) \sim \zeta \circ \phi^{\mathbb I}(\tau)$. We have $z_{\sharp}\circ \sigma = \omega$ and $z_{\sharp}\circ \tau = \nu$. Hence  $f^{\mathbb I}(\omega) = f^{\mathbb I}(z_{\sharp}\circ \sigma) = f^{\mathbb I}(|z_{\sharp}|^{\mathbb I}( \sigma)) = (f\circ |z_{\sharp}|)^{\mathbb I}(\sigma) = (|\zeta|\circ \phi)^{\mathbb I}(\sigma) = |\zeta|^{\mathbb I}(\phi^{\mathbb I}(\sigma)) = \zeta \circ  \phi^{\mathbb I}(\sigma)$ and similarly $f^{\mathbb I}(\nu) = \zeta \circ  \phi^{\mathbb I}(\tau)$. The result follows.
\end{proof}

\begin{theor} \label{weakmorTL}
Let $\A = (P,I,F, \lambda)$ and $\B = (P',I',F', \lambda')$ be two $M$-HDAs such that $\B$ is stable and there exists a weak morphism $f\colon \A \to \B$. Then for any two labels $l, m \in L(\A)$, $l\equiv_{\A} m \Rightarrow l \equiv_{\B} m$ and a map $TL(\A) \to TL(\B)$ is given by $tr_{\A}(l) \mapsto tr_{\B}(l)$. Consequently, if $\A$ and $\B$ are stable $M$-HDAs and $\A \leftrightarrow \B$, then $TL(\A) = TL(\B)$.
\end{theor}

\begin{proof}
Let $l,m \in L(\A)$ such that $l\equiv_{\A} m$. By \ref{labellem}, there exist dihomotopic paths $\omega, \nu \in P^{\mathbb I}$ beginning in $I$ and ending in $F$ such that $\overline{\lambda} (\omega) = l$ and $\overline{\lambda} (\nu) = m$. By \ref{fadj},  $f^{\mathbb I}(\omega) \sim f^{\mathbb I}(\nu)$. Since $\B$ is stable, $\overline{\lambda'} (f^{\mathbb I}(\omega)) \equiv_{\B} \overline{\lambda'} (f^{\mathbb I}(\nu))$. The result follows.
\end{proof}

\subsection{The trace category of an HDA}
The \emph{fundamental category} of a precubical set $P$ is the category $\vec \pi_1(P)$ whose objects are the vertices of $P$ and whose morphisms are the dihomotopy classes of paths in $P$ (cp. \cite{Goubault, GrandisBook}). Given a weak morphism of precubical sets $f\colon |P| \to |Q|$, the functor $f_*\colon \vec \pi_1(P) \to \vec \pi_1(Q)$ is defined  by $f_*(v) = f_0(v)$ $(v\in P_0)$ and  $f_*([\omega]) = [f^{\mathbb I}(\omega)]$ $(\omega \in P^{\mathbb I})$. It is clear that the fundamental category construction is functorial. 

A vertex $v$ of a precubical set $P$ is said to be \emph{maximal (minimal)} if there is no element $x \in P_1$ such that $d_1^0x = v$ $(d_1^1x = v)$. The sets of maximal and minimal elements of $P$ are denoted by $M(P)$ and $m(P)$ respectively. The \emph{trace category} of an $M$-HDA $\A = (P,I,F,\lambda)$, $TC(\A)$, is the full subcategory of $\vec \pi_1(P)$ generated by $I\cup F\cup m(P) \cup M(P)$. Our definition of the trace category of an HDA is a variant  of Bubenik's definition of the fundamental bipartite graph of a d-space \cite{BubenikExtremal}.

Note that if $f$ is a weak morphism from an $M$-HDA $\A = (P,I, F,\lambda)$ to an $M\textrm{-}$HDA $\B = (P',I', F',\lambda')$ such that $f_0(m(P) \cup M(P)) \subseteq I'\cup F'\cup m(P') \cup M(P')$, then the functor $f_*\colon \vec \pi_1(P) \to \vec \pi_1(Q)$ restricts to a functor $f_*\colon TC(\A) \to TC(\B)$.

\subsection{Accessible HDAs} \label{accessibleHDAs} An $M$-HDA $\A = (P,I,F,\lambda)$ is said to be \emph{accessible} if for every vertex $v\in P_0$ there exists a path in $P$ from an initial state to $v$ (cp. \cite[def. I.1.10]{Sakarovitch}). If $\A$ and $\B$ are accessible, deterministic and stable $M$-HDAs and $\A \leftrightarrow \B$, then $L(\A) = L(\B)$, $TL(\A) = TL(\B)$ and $TC(\A) \cong TC(\B)$. This follows from \ref{lalb}, \ref{weakmorTL} and the following proposition, which shows that in the case of accessible and deterministic HDAs, $\A \leftrightarrow \B$ is a very strong condition:

\begin{prop}
Let $\A = (P,I,F,\lambda)$ and $\B=(P',I',F',\lambda')$ be two accessible and deterministic $M$-HDAs and $f\colon \A \to \B$ and $g\colon \B \to \A$ be weak morphisms. Then 
\begin{enumerate}[(i)]
\item $f_0\colon P_0 \to P'_0$ and $g_0\colon P'_0 \to P_0$ are inverse bijections;
\item $f^{\mathbb I}\colon P^{\mathbb I} \to P'^{\mathbb I}$ and $g^{\mathbb I}\colon P'^{\mathbb I} \to P^{\mathbb I}$ are inverse bijections;
\item $f_0$ and $g_0$ preserve maximal and minimal elements;
\item $f_* \colon \vec \pi_1(P) \to \vec \pi_1(P')$ and $g_* \colon \vec \pi_1(P') \to\vec \pi_1(P)$ are inverse isomorphisms;
\item $f_*\colon TC(\A) \to TC(\B)$ and $g_*\colon TC(\B) \to TC(\A)$ are inverse isomorphisms.
\end{enumerate}
\end{prop}

\begin{proof} It is enough to show (i), (ii) and (iii). 

(i) Since $\A$ and $\B$ are deterministic, there exist $a \in P_0$ and $a' \in P'_0$ such that  $I = \{a\}$ and $I' = \{a'\}$. We have $f_0(a) = a'$ and $g_0(a') = a$. Consider a path $\omega\in P^{\mathbb I}$ such that $\omega (0) = a$. Then $g^{\mathbb I}(f^{\mathbb I}(\omega))(0) = g_0(f^{\mathbb I}(\omega) (0)) = g_0(f_0(\omega (0))) = g_0(f_0(a)) = a = \omega (0)$. Since $\overline{\lambda}(g^{\mathbb I}(f^{\mathbb I}(\omega))) = \overline{\lambda}(\omega)$ and $\A$ is deterministic, $g^{\mathbb I}(f^{\mathbb I}(\omega)) = \omega$. Analogously, $f^{\mathbb I}(g^{\mathbb I}(\omega')) = \omega'$ for any path $\omega ' \in P'^{\mathbb I}$ with $\omega'(0) = a'$.

Let $v \in P_0$ be any vertex. Since $\A$ is accessible, there exists a path $\omega \colon \llbracket 0,k \rrbracket \to P$ such that $\omega (0) = a$ and $\omega (k) = v$. As we have seen, $g^{\mathbb I}(f^{\mathbb I}(\omega)) = \omega$. Suppose that $\length (f^{\mathbb I}(\omega)) = k'$. We have $g_0(f_0(v)) = g_0(f_0(\omega (k))) = g_0(f^{\mathbb I}(\omega) (k')) = g^{\mathbb I}(f^{\mathbb I}(\omega)) (k) = \omega (k) = v$. Similarly, $f_0(g_0(v')) = v'$ for all vertices $v' \in P'_0$. This shows that $f_0$ and $g_0$ are inverse bijections.

(ii) Let $\omega \in P^{\mathbb I}$ be any path. Then $g^{\mathbb I}(f^{\mathbb I}(\omega))(0) = g_0(f_0(\omega (0)) = \omega (0)$. Since $\overline{\lambda}(g^{\mathbb I}(f^{\mathbb I}(\omega))) = \overline{\lambda}(\omega)$ and $\A$ is deterministic, $g^{\mathbb I}(f^{\mathbb I}(\omega)) = \omega$. Analogously, \linebreak $f^{\mathbb I}(g^{\mathbb I}(\omega')) = \omega'$ for any path $\omega ' \in P'^{\mathbb I}$. This shows that $f^{\mathbb I}$ and $g^{\mathbb I}$ are inverse bijections. 

(iii) It is enough to show that $f_0$ preserves maximal and minimal elements.. Let $v \in P_0$ be an element such that $f_0(v)$ is not maximal. Then there exists an element $x'\in P'_1$ such that $x'_{\sharp}(0) = d_1^0x' = f_0(v)$. Then $g^{\mathbb I}(x'_{\sharp})(0) = g_0(x'_{\sharp}(0)) = g_0(f_0(v)) = v$. Since $\length(g^{\mathbb I}(x'_{\sharp})) > 0$, $v$ is not maximal. This shows that $f_0$ preserves maximal elements. Similarly, $f_0$ preserves minimal elements. 
\end{proof}

\subsection{Trace equivalences} \label{simt} 

A \emph{trace equivalence} from an $M$-HDA $\A = (P,I,F,\lambda)$ to an $M$-HDA $\B=(P',I',F',\lambda')$ is a weak morphism $f$ from $\A$ to $\B$ such that $f_0(I) = I'$, $f_0(F) = F'$, $f_0(m(P)) = m(P')$, $f_0(M(P)) = M(P')$ and the functor  $f_*\colon TC(\A) \to TC(\B)$ is an isomorphism. We write $\A \stackrel{\sim_t}{\to} \B$ and say that $\A$ is a \emph{trace equivalent abstraction} of $\B$, or that $\B$ is a \emph{trace equivalent refinement} of $\A$, if there exists a trace equivalence from $\A$ to $\B$. It is clear that $\stackrel{\sim_t}{\to}$ is a preorder on the class of $M$-HDAs. We have seen in \ref{accessibleHDAs} that for accessible and deterministic $M$-HDAs $\A$ and $\B$, $\A \leftrightarrow \B \Rightarrow \A \stackrel{\sim_t}{\to} \B$. In the next section, we give another condition under which one has $\A \stackrel{\sim_t}{\to} \B$. Here, we show that in the stable and deterministic case, the trace language of an $M$-HDA is invariant under trace equivalent abstraction:

\begin{theor}\label{traceequivalence}
Let $\A$ be a stable $M$-HDA and $\B$ be a stable and deterministic $M$-HDA such that $\A \stackrel{\sim_t}{\to} \B$.  Then a bijection $TL(\A) \to TL(\B)$ is given by $tr_{\A}(l) \mapsto tr_{\B}(l)$. 
\end{theor}

\begin{proof}
Write $\A  = (P,I,F,\lambda)$ and $\B=(P',I',F',\lambda')$, and let $f$ be a trace equivalence from $\A$ to $\B$. By theorem \ref{weakmorTL}, the map $\Psi\colon TL(\A) \to TL(\B)$, $tr_{\A}(l) \mapsto tr_{\B}(l)$ is well-defined. 

We show that $\Psi$ is surjective. Consider $l' \in L(\B)$. Let $\omega'$ be a path in $P'$ beginning in $I'$ and ending in $F'$ such that $\overline{ \lambda'} (\omega ') = l'$. Since $f_*\colon TC(\A) \to TC(\B)$ is an isomorphism, there exists a path $\omega$  in $P$ with end points in $I \cup F\cup m(P) \cup M(P)$ such that $f_*([\omega]) = [\omega ']$. It follows that $f^{\mathbb I}(\omega ) \sim \omega'$ and that $f^{\mathbb I}(\omega )$ begins in $I'$ and ends in $F'$. Since the functor $f_*\colon TC(\A) \to TC(\B)$ is an isomorphism, the map $f_0\colon I \cup F \cup m(P) \cup M(P) \to I'\cup F' \cup m(P') \cup M(P')$ is a bijection. Since $f_0(I) = I'$ and $f_0(F) = F'$, it follows that $\omega$ begins in $I$ and ends in $F$. Therefore $\overline{\lambda} (\omega ) \in L(\A)$. By \ref{stablelabel}, $l' = \overline{ \lambda'} (\omega ') \equiv_{\B} \overline{ \lambda'} (f^{\mathbb I}(\omega )) = \overline{ \lambda} (\omega)$. Thus, $tr_{\B}(l') = tr_{\B}(\overline{ \lambda} (\omega))  = \Psi (tr_{\A}(\overline{ \lambda} (\omega)))$. 

It remains to show that $\Psi$ is injective. Consider $l,m \in L(\A)$ such that $tr_{\B}(l) = tr_{B}(m)$, i.e. $l \equiv_{\B} m$. Let $\omega$ and $\nu$ be paths in $P$ beginning in $I$ and ending in $F$ such that $\overline{ \lambda} (\omega) = l$ and $\overline{ \lambda} (\nu) = m$. Then $f^{\mathbb I}(\omega )$ and $f^{\mathbb I}(\nu)$ start in the only initial state of $\B$ and $\overline{ \lambda'} (f^{\mathbb I}(\omega)) = l \equiv_{\B} m = \overline{ \lambda'} (f^{\mathbb I}(\nu))$. By \ref{weakmordet}, $f^{\mathbb I}(\omega) \sim f^{\mathbb I}(\nu)$. Thus, $f_*([\omega]) = f_*([\nu])$. Since $f_*$ is an isomorphism, $\omega \sim \nu$. By \ref{stablelabel}, $l \equiv_{\A} m$, i.e. $tr_{\A}(l) = tr_{\A}(m)$.  
\end{proof}

\section{Homeomorphic abstraction} \label{homeoabs}

In this last section, we introduce the preorder relation of homeomorphic abstraction and show that, under a mild condition, it is stronger than trace equivalent abstraction. The main point is to construct an inverse ``up to dihomotopy'' of the map induced on paths sets by a weak morphism that is a homeomorphism. The construction of this inverse relies on cube paths and carrier sequences.

\subsection{The preorder relation of homeomorphic abstraction} \label{homeoref} Consider two $M\textrm{-}$HDAs $\A = (P,I,F,\lambda)$ and $\B = (P',I',F',\lambda')$. We say that $\A$ is a \emph{homeomorphic abstraction} of $\B$, or that $\B$ is a \emph{homeomorphic refinement} of $\A$, if there exists a weak morphism $f$ from $\A$ to $\B$ that is a homeomorphism and satisfies $f_0(I) = I'$ and $f_0(F) = F'$. In particular, if $\B$ is a subdivision of $\A$, then $\A$ is a homeomorphic abstraction of $\B$. We use the notation $\A \stackrel{\approx}{\to} \B$ to indicate that $\A$ is a homeomorphic abstraction of $\B$. It is clear that the relation $\stackrel{\approx}{\to}$ is a preorder on the class of $M$-HDAs.

\subsection{Cube paths}
The cube paths we define next are essentially the same as the full cube paths considered in \cite{FahrenbergThesis} and the paths considered in \cite{vanGlabbeek}. A \emph{cube path} in a precubical set $P$ is a sequence of elements
$c = (c_1, \dots, c_m)$ such that for all $i \in \{1,\dots,m-1\}$ one of the following conditions holds:
\begin{itemize}
\item[(i)] $c_i = c_{i+1}$;
\item[(ii)] $\exists r\in \{1, \dots , \deg(c_i)\}\colon c_{i+1} = d_r^1c_i$;
\item[(iii)] $\exists r\in \{1, \dots , \deg(c_{i+1})\}\colon c_i = d_r^0c_{i+1}$.
\end{itemize}
The \emph{concatenation} of two cube paths $c = (c_1, \dots , c_m)$ and $d = (d_1, \dots , d_n)$ such that $d_1 = c_m$ is the cube path $c\cdot d = (c_1, \dots , c_m=d_1, \dots , d_n)$. If $\chi\colon  Q \to P$ is a morphism of precubical sets and $c = (c_1, \dots, c_m)$ is a cube path in $Q$, then we denote by $\chi(c)$ the cube path $(\chi(c_1), \dots, \chi(c_m))$ in $P$.  

With every cube path we can associate an ordinary path. In the definition of this path we use the following notation: Let $b \in P$ be an element of degree $n > 0$ and $r \in \{1, \dots, n\}$. We define the element $\hat d_rb \in P_1$ by 
$$\hat d_rb = \left\{\begin{array}{ll} b,& n = 1,\\d_1^0 \cdots d_{r-1}^ 0d_{r+1}^ 0\cdots d_n^ 0b, & n > 1. \end{array} \right.$$ Note that the element $\hat d_rb$ is an edge of $b$ leading from the initial vertex of $b$ to the initial vertex of the face $d_r^ 1b$, i.e. we have $d_1^ 0\hat d_rb = d_1^ 0\cdots d_1^0b$ and $d_1^1\hat d_rb = d_1^ 0\cdots d_1^0d_r^ 1b$.

\begin{defin}
Let $c=(c_1, \dots, c_m)$ be a cube path in $P$. The \emph{path associated with} $c$ is the path $\gamma(c) \in P^{\mathbb I}$ defined recursively as follows: If either $m = 1$ or $m=2$ and $\deg(c_1) \leq \deg(c_2)$, then $\gamma(c)$ is the path of length $0$ defined by $\gamma(c)(0) = c_{1\sharp}(0,\dots,0)$. If $m=2$ and $\deg(c_1) > \deg(c_2)$, then consider the least $r \in \{1, \dots, \deg(c_1)\}$ such that $c_2 = d_r^1c_1$ and set $\gamma(c) = (\hat d_rc_1)_{\sharp}$. For $m > 2$, $\gamma(c) = \gamma (c_1,c_2)\cdot \gamma(c_2,c_3) \cdots \gamma (c_{m-1},c_m)$.
\end{defin}

\begin{rems}
(i) We remark that $\gamma (c)$ is a path in $P$ from $c_{1\sharp}(0,\dots ,0)$ to $c_{m\sharp}(0,\dots ,0)$.

(ii) We have $\gamma(c\cdot d) = \gamma(c) \cdot \gamma(d)$.
\end{rems}

\subsection{Carrier sequences}
Let $f\colon |P| \to |P'|$ be a weak morphism of precubical sets that is a homeomorphism. We shall show that with every path in $P'$ we can associate a cube path, which we call its carrier sequence. The concept of carrier sequence is an adaptation of the one considered in \cite{Fajstrup}.

\begin{defin} 
The \emph{carrier} of an element $b \in P'$ with respect to $f$ is the unique element $c(b)\in P$ for which there exists an element $u \in ]0,1[^{\deg(c(b))}$ such that $f([c(b),u]) = [b,(\tfrac{1}{2},\dots,\tfrac{1}{2})]$.  The \emph{carrier sequence} of a path $\nu = y_{1\sharp}\cdots y_{n\sharp}$ in $P'$ $(n\geq 1)$ with respect to $f$ is the sequence
$$c(\nu) = \left(c(d_1^0y_1), c(y_1),c(d_1^1y_1),c(y_2),c(d_1^1y_2),\dots , c(y_n),c(d_1^1y_n)\right).$$ For a path $\nu \in P'^{\mathbb I}$ of length $0$ we set $c(\nu) = \left( c(\nu(0)) \right)$.
\end{defin}

\begin{rem}  \label{cf0}
For all $a \in P_0$, $c(f_0(a)) =a$. Indeed, $f([a,()]) = [f_0(a),()]$.
\end{rem}

\begin{prop}
The carrier sequence of a path in $P'$ is a cube path in $P$.
\end{prop}

\begin{proof}
The proposition is an immediate consequence of lemma \ref{cdb} below.
\end{proof}

\begin{rems}
(i) Since the carrier sequence of a path $\omega \in P'^{\mathbb I}$ is a cube path in $P$, we may consider the path $\gamma (c (\omega))$ in $P$. If $v$ and $w$ are vertices in $P$ and $\omega$ is a path in $P'$ from $f_0(v)$ to $f_0(w)$, then $\gamma(c(\omega))$ is a path in $P$ from $v$ to $w$.
 
(ii) If $\omega$ and $\nu$ are paths in $P'$ such that $\omega(\length(\omega )) = \nu (0)$, then $c(\omega \cdot \nu) = c(\omega) \cdot c(\nu)$. 
\end{rems}

\begin{lem} \label{cdb}
Consider elements $b\in P'_1$ and $k \in \{0,1\}$. Then either $c(d_1^kb) =  c(b)$ or there exists an $r \in \{1, \dots, \deg(c(b))\}$ such that $c(d_1^kb) =  d_r^kc(b)$.
\end{lem}

\begin{proof}
Set $n = \deg(c(b))$. Note that $n\geq 1$ because otherwise we would have $[f_0(c(b)),()] = [b,\frac{1}{2}]$. Let $u \in ]0,1[^n$ be the element such that $f([c(b),u]) = [b,\frac{1}{2}]$. Since $f$ is a weak morphism, there exist integers $k_1, \dots, k_n \geq 1$, a morphism of precubical sets $\beta\colon \llbracket 0,{k_1} \rrbracket \otimes \cdots \otimes \llbracket 0,{k_n} \rrbracket \to P'$  and a dihomeomorphism $$\phi\colon |\llbracket 0,1 \rrbracket^{\otimes n}| = [0,1]^n \to |\llbracket 0,{k_1} \rrbracket \otimes \cdots \otimes \llbracket 0,{k_n} \rrbracket| = [0,k_1]\times \cdots \times [0,k_n]$$ such that $f\circ |c(b)_{\sharp}| = |\beta| \circ \phi$. Since $$[b,\tfrac{1}{2}] \in f(|c(b)_{\sharp}|([0,1]^n)) = |\beta (\llbracket 0,{k_1} \rrbracket\otimes \cdots \otimes \llbracket 0,{k_n} \rrbracket)|,$$ we have $b \in  \beta (\llbracket 0,{k_1} \rrbracket\otimes \cdots \otimes \llbracket 0,{k_n} \rrbracket)$. Write $$b = \beta(i_1, \dots, i_{j-1}, [i_j,i_j+1],i_{j+1}, \dots, i_n).$$ 
Since \begin{eqnarray*}
\lefteqn{f([c(b),\phi^{-1}(i_1, \dots, i_{j-1}, i_j+\tfrac{1}{2},i_{j+1}, \dots, i_n)])}\\ &=& |\beta| (i_1, \dots, i_{j-1}, i_j+\tfrac{1}{2},i_{j+1}, \dots, i_n)\\ &=& |\beta| ([(i_1, \dots, i_{j-1}, [i_j,i_j+1],i_{j+1}, \dots, i_n),\tfrac{1}{2}])\\ &=&  [b,\tfrac{1}{2}]\\ &=& f([c(b),u]), 
\end{eqnarray*}
we have $[c(b),\phi^{-1}(i_1, \dots, i_{j-1}, i_j+\tfrac{1}{2},i_{j+1}, \dots, i_n)] = [c(b),u]$. Since $u \in ]0,1[^n$, it follows that $\phi^{-1}(i_1, \dots, i_{j-1}, i_j+\tfrac{1}{2},i_{j+1}, \dots, i_n)= u$ and that $i_r \notin \{0,k_r\}$ for $r\not= j$. 

We have $d_1^0b = \beta(i_1, \dots,  i_n)$ and $f([c(b),\phi^ {-1}(i_1, \dots, i_n)]) = |\beta|(i_1, \dots, i_n) = [d_1^ 0b,()]$. If $i_j > 0$, then $\phi^ {-1}(i_1, \dots, i_n) \in ]0,1[^ n$ and therefore  $c(d_1^0b) = c(b)$. If $i_j=0$, then, by proposition  \ref{minmaxcoord}, exactly one coordinate of $\phi^ {-1}(i_1, \dots, i_n)$ is equal to $0$ and none is equal to $1$. This implies that there exist elements $r\in \{1, \dots, n\}$ and $v\in ]0,1[^{n-1}$ such that $\phi^ {-1}(i_1, \dots, i_n) = \delta_r^0(v)$. Therefore $[d_1^0b,()] = f([c(b),\delta_r^0(v)]) = f([d_r^0c(b),v])$ and hence $c(d_1^0b) = d_r^0c(b)$. An analogous argument shows that either $c(d_1^1b) = c(b)$ or there exists an index $r\in \{1, \dots, n\}$ such that $c(d_1^1b) = d_r^1c(b)$.
\end{proof}

\subsection{The map $\gamma \circ c$}
Consider a weak morphism of precubical sets $f\colon |P| \to |P'|$ that is a homeomorphism. We show that the map $\gamma \circ c\colon P'^{\mathbb I} \to P^{\mathbb I}$ is a retraction of $f^{\mathbb I}$:

\begin{prop} \label{retraction}
For every path $\omega \in P^{\mathbb I}$, $\gamma \circ c(f^{\mathbb I}(\omega)) = \omega$.
\end{prop}

\begin{proof}
Suppose first that $\length(\omega) =0$. Then $\length(f^{\mathbb I}(\omega)) =0$ and $f^{\mathbb I}(\omega)(0) = f_0(\omega(0))$. Therefore $c(f^{\mathbb I}(\omega)) = (c(f_0(\omega(0))) = (\omega(0))$. Thus, $\gamma (c(f^{\mathbb I}(\omega))) = \gamma (\omega (0))$. Since $\gamma (\omega (0))(0) = \omega(0)_{\sharp}(0) = \omega (0)$, we have $\gamma (c(f^{\mathbb I}(\omega))) = \omega$.

Suppose now that $\omega = x_{\sharp}$ for some $x \in P_1$. Suppose that $\length(f^{\mathbb I}(\omega)) = k$ and that $\phi\colon |\llbracket 0,1 \rrbracket| = [0,1] \to |\llbracket 0,k \rrbracket| = [0,k]$ is a dihomeomorphism such that $|f^{\mathbb I}(\omega)|\circ \phi = f \circ |\omega|$. Write $y_i = f^{\mathbb I}(\omega)([i-1,i])$. Then $f^{\mathbb I}(\omega) = y_{1\sharp} \cdots y_{k\sharp}$ and hence $c(f^{\mathbb I}(\omega)) = c(y_{1\sharp}) \cdots c(y_{k\sharp})$. For $i \in \{1, \dots, k\}$, $\phi^{-1}(i-\frac{1}{2}) \in ]0,1[$ and $f([x, \phi^{-1}(i-\frac{1}{2})]) = f\circ |\omega|(\phi^ {-1}(i-\frac{1}{2})) = |f^{\mathbb I}(\omega)|(i-\frac{1}{2}) = [f^{\mathbb I}(\omega)([i-1,i]),\frac{1}{2}]) = [y_i,\frac{1}{2}]$. Thus, $c(y_i) = x$. For $i \in \{1, \dots, k-1\}$, $\phi^{-1}(i) \in ]0,1[$ and $f([x, \phi^{-1}(i)]) = f\circ |\omega|(\phi^ {-1}(i)) = |f^{\mathbb I}(\omega)|(i) = [f^{\mathbb I}(\omega) (i),()] = [d_1^1y_i,()]$. Thus, $c(d_1^1y_i) = x$. We have $c(d_1^ 0y_1) = c(f^{\mathbb I}(\omega)(0)) = c(f_0(\omega(0))) = \omega(0) = d_1^ 0x$ and $c(d_1^ 1y_k) = c(f^{\mathbb I}(\omega)(k)) = c(f_0(\omega(1))) = \omega(1) = d_1^ 1x$. It follows that $c(f^{\mathbb I}(\omega)) = (d_1^0x, x, \dots , x, d_1^ 1x)$ and hence that $\gamma (c (f^{\mathbb I}(\omega)) = \gamma (d_1^ 0x,x) \cdot \gamma (x,x) \cdots \gamma(x,x) \cdot \gamma(x,d_1^ 1x) = (d_1^ 0x)_{\sharp} \cdots (d_1^ 0x)_{\sharp} \cdot x_{\sharp} = x_{\sharp} = \omega$.

Suppose finally that $\omega = x_{1\sharp} \cdots x_{k\sharp}$ where $(x_1, \dots , x_k)$ is a sequence of elements of $P_1$ such that $d_1^0x_{j+1} = d_1^1x_j$ for all $1\leq j < k$. Then $\gamma(c(f^{\mathbb I}(\omega))) = \gamma(c(f^{\mathbb I}(x_{1\sharp} \cdots x_{k\sharp}))) =  \gamma(c(f^{\mathbb I}(x_{1\sharp})))\cdots \gamma(c(f^{\mathbb I}(x_{k\sharp}))) = x_{1\sharp} \cdots x_{k\sharp} = \omega$.
\end{proof}

\subsection{Regular and weakly regular elements}

We say that an element $x$ of a precubical set is \emph{regular} if $x_{\sharp}$ is injective. We say that an element $x$ of degree $n$ is \emph{weakly regular} if the restrictions of $x_{\sharp}$ to the graded subsets $(\llbracket 0,1 \rrbracket\setminus \{0\})^ {\otimes n}$ and $(\llbracket 0,1 \rrbracket\setminus \{1\})^ {\otimes n}$ of $\llbracket 0,1 \rrbracket^ {\otimes n}$ are injective. A precubical set or an $M$-HDA is said to be \emph{(weakly) regular} if all of its elements are (weakly) regular. Every $0$- or $1$-dimensional precubical set is weakly regular. The \emph{directed circle}, which is the precubical set consisting of exactly one vertex and exactly one $1$-cube, is weakly regular but not regular. It can be shown that a precubical set is weakly regular if and only if all elements of degree $2$ are weakly regular. Since we do not need this result in this paper, we do not include a proof.

\subsection{Compatibility of $\gamma \circ c$ with dihomotopy} 

The purpose of this subsection is to show that the map $\gamma \circ c$ sends dihomotopic paths to dihomotopic paths, at least if we restrict ourselves to weakly regular precubical sets.

\begin{lem} \label{cnat}
Consider two weak morphisms of precubical sets $f\colon |P| \to |P'|$ and $g\colon |Q| \to |Q'|$ and two morphisms of precubical sets $\chi \colon Q \to P$ and $\chi'\colon Q'\to P'$ such that $f$ and $g$ are homeomorphisms and $f\circ |\chi| = |\chi'| \circ g$. Then for all $a \in Q'$, $c(\chi'(a)) = \chi (c(a))$.
\end{lem}

\begin{proof}
Consider $u \in ]0,1[^{\deg(c(a))}$ such that $g([c(a),u]) = [a,(\frac{1}{2}, \dots, \frac{1}{2})]$. We have $f([\chi(c(a)),u]) = f\circ |\chi|([c(a),u]) = |\chi'|\circ g([c(a),u]) = |\chi'|([a,(\frac{1}{2}, \dots, \frac{1}{2})]) = [\chi'(a),(\frac{1}{2}, \dots, \frac{1}{2})]$ and hence $c(\chi'(a)) = \chi(c(a))$. 
\end{proof}

\begin{lem} \label{weaklyreg}
Let $\chi\colon  Q \to P$ be a morphism of precubical sets such that $P$ is weakly regular. Then for any cube path $c$ in $Q$, $\gamma (\chi(c)) = \chi \circ \gamma(c)$.
\end{lem}

\begin{proof}
Let $c = (c_1, \dots, c_m)$ be a cube path in $Q$. Recall that $\chi (c)$ is the cube path $(\chi(c_1), \dots, \chi(c_m))$ in $P$. In order to proof the lemma, it is enough to consider the case $m \leq 2$. If either $m = 1$ or $m=2$ and $\deg(c_1) \leq \deg(c_2)$, then $\gamma(c)$ is the path in $Q$ of length $0$ defined by $\gamma(c)(0) = c_{1\sharp}(0,\dots,0)$. Hence $\chi \circ \gamma(c)$ is the path in $P$ of length $0$ given by $(\chi \circ \gamma(c))(0) = \chi(c_{1\sharp}(0,\dots,0)) = \chi(c_1)_{\sharp}(0,\dots,0)$. Thus, $\chi \circ \gamma(c) = \gamma (\chi(c))$. Suppose that $m= 2$ and  $\deg(c_1) > \deg(c_2)$. Consider the least $r \in \{1, \dots, \deg(c_1)\}$ such that $c_2 = d_r^1c_1$. We have $\gamma (c) = (\hat d_rc_1)_{\sharp}$ and hence $\chi \circ \gamma (c) = \chi \circ (\hat d_rc_1)_{\sharp} = \chi (\hat d_rc_1)_{\sharp}$. Moreover, $\chi(c_2) = \chi(d_r^1c_1) = d_r^1(\chi(c_1))$. Since $P$ is weakly regular, this is the only such $r$ and $\gamma (\chi(c)) = (\hat d_r\chi(c_1))_{\sharp} = \chi (\hat d_rc_1)_{\sharp} = \chi \circ \gamma (c)$.
\end{proof}

\begin{prop} \label{gammac}
Let $P$ be a weakly regular precubical set and $f\colon |P| \to |P'|$ be a weak morphism that is a homeomorphism. Then for any two paths $\omega, \nu \in P'^{\mathbb I}$, $\omega \sim \nu \Rightarrow \gamma \circ c(\omega) \sim \gamma \circ c(\nu)$.
\end{prop}

\begin{proof}
We may suppose that $\omega$ and $\nu$ are elementarily dihomotopic paths of length $2$. We may further suppose that there exists an element $y \in P'_2$ such that $\omega = (d_1^0y)_{\sharp}\cdot (d_2^1y)_{\sharp}$ and $\nu = (d_2^0y)_{\sharp}\cdot (d_1^1y)_{\sharp}$. Set $x = c(y)$ and $n = \deg(x)$. Since $f$ is a weak morphism, there exist a morphism of precubical sets $\beta\colon \llbracket 0,{k_1} \rrbracket \otimes \cdots \otimes \llbracket 0,{k_n} \rrbracket \to P'$ and a dihomeomorphism $$\phi\colon |\llbracket 0,1 \rrbracket^{\otimes n}| = [0,1]^n \to |\llbracket 0,{k_1} \rrbracket \otimes \cdots \otimes \llbracket 0,{k_n} \rrbracket| = [0,k_1]\times \cdots \times [0,k_n]$$ such that $f\circ |x_{\sharp}| = |\beta| \circ \phi$. Since $$[y,(\tfrac{1}{2},\tfrac{1}{2})] \in f(|x_{\sharp}|([0,1]^n)) = |\beta (\llbracket 0,{k_1} \rrbracket\otimes \cdots \otimes \llbracket 0,{k_n} \rrbracket)|,$$ we have $y \in  \beta (\llbracket 0,{k_1} \rrbracket\otimes \cdots \otimes \llbracket 0,{k_n} \rrbracket)$. Consider $z\in \llbracket 0,{k_1} \rrbracket\otimes \cdots \otimes \llbracket 0,{k_n} \rrbracket$ such that $\beta (z) = y$. Consider the paths $\xi = (d_1^0z)_{\sharp}\cdot (d_2^1z)_{\sharp}$ and $\theta = (d_2^0z)_{\sharp}\cdot (d_1^1z)_{\sharp}$. Since $\xi$ and $\theta$ have the same end points, so do $\gamma (c(\xi))$ and $\gamma (c(\theta))$. Since any two paths in $\llbracket 0,1 \rrbracket^{\otimes n}$ with the same end points are dihomotopic (see \ref{detcube}), we have $\gamma (c(\xi))\sim \gamma (c(\theta))$ and hence $x_{\sharp} \circ \gamma (c(\xi))\sim x_{\sharp} \circ \gamma (c(\theta))$. Since $P$ is weakly regular, lemma \ref{weaklyreg} implies that $\gamma (x_{\sharp}(c(\xi)))\sim \gamma (x_{\sharp}(c(\theta)))$. By lemma \ref{cnat}, we have \begin{eqnarray*}
x_{\sharp}(c(\xi)) &=& x_{\sharp}(c((d_1^0z)_{\sharp}\cdot (d_2^1z)_{\sharp}))\\ &=& x_{\sharp}(c(d_1^0d_1^0z),c(d_1^0z),c(d_1^1d_1^0z),c(d_2^1z),c(d_1^1d_2^1z))\\ &=& (x_{\sharp}(c(d_1^0d_1^0z)),x_{\sharp}(c(d_1^0z)),x_{\sharp}(c(d_1^1d_1^0z)),x_{\sharp}(c(d_2^1z)),x_{\sharp}(c(d_1^1d_2^1z)))\\ &=& (c(\beta(d_1^0d_1^0z)),c(\beta(d_1^0z)),c(\beta(d_1^1d_1^0z)),c(\beta(d_2^1z)),c(\beta(d_1^1d_2^1z)))\\ &=&  (c(d_1^0d_1^0y),c(d_1^0y),c(d_1^1d_1^0y),c(d_2^1y),c(d_1^1d_2^1y))\\ &=& c(\omega).\end{eqnarray*}
Similarly, $x_{\sharp}(c(\theta)) = c(\nu)$. Thus, $\gamma (c(\omega)) \sim \gamma (c(\nu))$.
\end{proof}

\subsection{The composite $f^{\mathbb I} \circ \gamma \circ c$} 
Let $f\colon |P| \to |P'|$ be a weak morphism of precubical sets that is a homeomorphism. 
We have seen in \ref{retraction} that the composite $\gamma \circ c \circ f^{\mathbb I}$ is the identity. Here, we show that ``up to dihomotopy'', the  composite $f^{\mathbb I} \circ \gamma \circ c$ is the identity for paths with end points in $f_0(P_0)$. We need three lemmas.

\begin{lem} \label{bbu}
Consider elements $b \in \llbracket 0,{k_1} \rrbracket\otimes \cdots \otimes \llbracket 0,{k_m} \rrbracket$ $(m, k_1, \dots, k_m \geq 1)$ and $u \in ]0,1[^ {\deg(b)}$. Then $b \in  (\llbracket 0,{k_1} \rrbracket\setminus \{k_1\}) \otimes \cdots \otimes (\llbracket 0,{k_m} \rrbracket\setminus \{k_m\})$ if and only if $[b,u] \in  [0,k_1[ \times \cdots \times [0,k_m[$.
\end{lem}

\begin{proof}
Write $b = (b_1, \dots, b_m)$ and suppose $\deg(b) = p$. Then there exist indices $1 \leq i_1 < \dots < i_p \leq m$ such that $\deg(b_{i_q}) = 1$ for $q \in \{1, \dots ,p\}$ and $b_i \in \{0, \dots, k_i\}$ for $i \notin \{i_1, \dots ,i_p\}$. For each $q \in \{1, \dots ,p\}$ there exists an element $j_q \in \{0,\dots, k_q-1\}$ such that $b_{i_q} = [j_q,j_q+1]$.  We have $[b,u] = [b,(u_1,\dots,u_p)] = (t_1, \dots, t_m) \in [0,k_1] \times \cdots \times [0,k_m]$ where $$t_i = \left\{ \begin{array}{ll} b_i,& i \notin \{i_1, \dots ,i_p\},\\ j_q+u_q, & q \in \{1, \dots ,p\}.\end{array}\right.$$
The result follows.
\end{proof}

\begin{lem} \label{bcb}
Let $\phi \colon |\llbracket 0,{l_1} \rrbracket\otimes \cdots \otimes \llbracket 0,{l_m} \rrbracket| \to |\llbracket 0,{k_1} \rrbracket\otimes \cdots \otimes \llbracket 0,{k_m} \rrbracket|$ $(k_i,l_i \geq 1)$ be a weak morphism of precubical sets that is a dihomeomorphism. Then for any element $b \in \llbracket 0,{k_1} \rrbracket\otimes \cdots \otimes \llbracket 0,{k_m} \rrbracket  $, $b \in  (\llbracket 0,{k_1} \rrbracket\setminus \{k_1\}) \otimes \cdots \otimes (\llbracket 0,{k_m} \rrbracket\setminus \{k_m\})$ if and only if $c(b) \in  (\llbracket 0,{l_1} \rrbracket\setminus \{l_1\}) \otimes \cdots \otimes (\llbracket 0,{l_m} \rrbracket\setminus \{l_m\})$.
\end{lem}

\begin{proof}
Consider an element $b \in \llbracket 0,{k_1} \rrbracket\otimes \cdots \otimes \llbracket 0,{k_m} \rrbracket  $, and let $u \in ]0,1[^{\deg(c(b))}$ be the element such that  $\phi([c(b),u]) = [b,(\frac{1}{2}, \dots, \frac{1}{2})]$. By proposition \ref{minmaxcoord}, $[b,(\frac{1}{2}, \dots, \frac{1}{2})] \in [0,k_1[ \times \cdots \times [0,k_m[$ if and only if $[c(b),u] \in [0,l_1[ \times \cdots \times [0,l_m[$. The preceding lemma implies the result.
\end{proof}

\begin{lem} \label{betainj}
Suppose that $P$ is weakly regular. Let $x \in P_n$ $(n \geq 1)$ be an element, $\beta\colon \llbracket 0,{k_1} \rrbracket\otimes \cdots \otimes \llbracket 0,{k_n} \rrbracket \to P'$ $(k_1, \dots, k_n \geq 1)$ be a morphism of precubical sets and $\phi\colon |\llbracket 0,1 \rrbracket^{\otimes n}| \to |\llbracket 0,{k_1} \rrbracket\otimes \cdots \otimes \llbracket 0,{k_n} \rrbracket|$ be a dihomeomorphism such that $f\circ |x_{\sharp}| = |\beta| \circ \phi$. Then the restriction of $\beta$ to  $(\llbracket 0,{k_1} \rrbracket\setminus \{k_1\})\otimes \cdots \otimes (\llbracket 0,{k_n} \rrbracket\setminus \{k_n\})$ is injective.
\end{lem}

\begin{proof}
Let $a,b \in (\llbracket 0,{k_1} \rrbracket\setminus \{k_1\})\otimes \cdots \otimes (\llbracket 0,{k_n} \rrbracket\setminus \{k_n\})$ be distinct elements of the same degree. Then $[a,(\tfrac{1}{2}, \dots, \tfrac{1}{2})]\not=  [b,(\tfrac{1}{2}, \dots, \tfrac{1}{2})] \in [0,k_1[\times \cdots \times [0,k_n[$ and therefore, by \ref{minmaxcoord}, 
$\phi^{-1}([a,(\tfrac{1}{2}, \dots, \tfrac{1}{2})])\not= \phi^ {-1}([b,(\tfrac{1}{2}, \dots, \tfrac{1}{2})]) \in [0,1[^n.$ Let $u \in ]0,1[^{\deg(c(a))}$ and $v \in ]0,1[^{\deg(c(b))}$ be the uniquely determined elements such that $\phi ([c(a),u]) = [a,(\tfrac{1}{2}, \dots, \tfrac{1}{2})]$ and $\phi ([c(b),v]) = [b,(\tfrac{1}{2}, \dots, \tfrac{1}{2})]$. 
By \ref{bcb}, $c(a), c(b) \in (\llbracket 0,1 \rrbracket\setminus\{1\})^ {\otimes n}$. If $c(a) = c(b)$, then $u \not=v$ and hence $|x_{\sharp}|([c(a),u])\not= |x_{\sharp}|([c(b),v])$. If $c(a) \not= c(b)$, then this holds because $x$ is weakly regular. Since $f$ is injective, it follows that $|\beta|([a,(\tfrac{1}{2}, \dots, \tfrac{1}{2})])\not= |\beta|([b,(\tfrac{1}{2}, \dots, \tfrac{1}{2})])$ and hence that $\beta(a) \not= \beta(b)$.
\end{proof}

\begin{prop} \label{fJgammac}
Suppose that $P$ is weakly regular. Let $v,w \in P_0$ be vertices and $\omega$ be a path in $P'$ from $f_0(v)$ to $f_0(w)$. Then $f^{\mathbb I}(\gamma \circ c(\omega)) \sim \omega$. 
\end{prop}

\begin{proof}
It is enough to show that there exists a path $\alpha \in P^{\mathbb I}$  such that  $f^{\mathbb I}(\alpha) \sim \omega$. Indeed, given such a path $\alpha$, we have, by \ref{retraction} and \ref{gammac}, $\alpha = \gamma\circ c (f^{\mathbb I}(\alpha)) \sim \gamma \circ c(\omega)$ and hence, by \ref{fadj}, $f^{\mathbb I}(\gamma \circ c(\omega)) \sim f^{\mathbb I}(\alpha) \sim \omega$.

In order to construct $\alpha$, we proceed by induction on the length of $\omega$. Suppose first that $\length(\omega) = 0$. Consider the path $\alpha \colon \llbracket 0,0 \rrbracket \to P$ defined by $\alpha (0) = v$. Then $f^{\mathbb I}(\alpha) $ is the path of length $0$ given by $f^{\mathbb I}(\alpha) (0) = f_0(\alpha(0)) = f_0(v) = \omega(0)$. Hence $f^{\mathbb I}(\alpha) = \omega$. It follows that the assertion holds for paths of length $0$. 

Suppose that $\length(\omega ) = n >0$ and that the assertion holds for every path in $P'$ that has length $< n$ and end points in $f_0(P_0)$. Write $\omega = y_{1\sharp}\cdots y_{n\sharp}$. Since $d_1^1y_n = \omega (n) = f_0(w)$, we have $c(d_1^1y_n) = w$. It follows that $\deg(c(d_1^1y_n)) = 0 < 1 = \deg(c(y_n))$. Let $l$ be the lowest index such that $\deg (c(d_1^1y_l)) < \deg (c(y_l))$. Consider the path $\nu = y_{1\sharp}\cdots y_{l\sharp}$, and let $\sigma$ be the unique path in $P'$ such that $\omega = \nu\cdot \sigma$. Set $x = c(y_l)$ and $m = \deg(x)$.  We shall show that there exist an index $r \in \{1, \dots, m\}$ and a path $\xi$ in $P'$ such that  $\nu \sim f^{\mathbb I}((\hat d_rx)_{\sharp})\cdot \xi$. This permits us to terminate the induction as follows: Since $\xi\cdot \sigma$ is a path of length $< n$ from $f_0(d_1^0\cdots d_1^0d_r^1x)$ to $f_0(w)$, the inductive hypothesis implies that there exists a path $\tau \in P^{\mathbb I}$ such that $f^{\mathbb I}(\tau) \sim \xi \cdot \sigma$. Since $f_0(\tau(0)) = f^{\mathbb I}(\tau)(0) = f_0(d_1^0\cdots d_1^0d_r^1x)$ and $f$ is a homeomorphism, $\tau (0) = d_1^0\cdots d_1^0d_r^1x$ and we may set $\alpha = (\hat d_rx)_{\sharp}\cdot \tau$. We then have $f^{\mathbb I}(\alpha) = f^{\mathbb I}((\hat d_rx)_{\sharp}) \cdot f^{\mathbb I}(\tau) \sim  f^{\mathbb I}((\hat d_rx)_{\sharp}) \cdot \xi\cdot \sigma \sim \nu \cdot \sigma = \omega$.

It remains to determine $r$ and $\xi$. Let $\beta\colon \llbracket 0,{k_1} \rrbracket\otimes \cdots \otimes \llbracket 0,{k_m} \rrbracket \to P'$ be the unique morphism of precubical sets and $\phi\colon |\llbracket 0,1 \rrbracket^{\otimes m}| \to |\llbracket 0,{k_1} \rrbracket\otimes \cdots \otimes \llbracket 0,{k_m} \rrbracket|$ be the unique dihomeomorphism such that $f\circ |x_{\sharp}| = |\beta| \circ \phi$.
We construct a path $\theta\colon \llbracket 0,l \rrbracket \to \llbracket 0,{k_1} \rrbracket\otimes \cdots \otimes \llbracket 0,{k_m} \rrbracket$ such that $\theta (0) = (0, \dots,0)$ and $|\beta|^{\mathbb I}(\theta) = \beta \circ \theta =  \nu$. For 
$i \in \{1,\dots, l\}$, by \ref{cdb}, we have either $c(d_1^0y_i) = c(y_i)$ or $c(d_1^0y_i) = d_q^0c(y_i)$ for some $q \in \{1, \dots, \deg(c(y_i))\}$. For $i < l$, by the definition of $l$, $c(y_i) = c(d_1^1y_i) = c(d_1^0y_{i+1})$. Hence for $i<l$, $c(y_i) = c(y_{i+1})$ or $c(y_i) = d_q^0c(y_{i+1})$ for some $q \in \{1, \dots, \deg(c(y_{i+1}))\}$. It follows that 
$c(y_i) \in x_{\sharp}((\llbracket 0,1 \rrbracket\setminus \{1\})^{\otimes m})$ for all $i \in \{1,\dots,l\}$. For $i \in \{1,\dots,l\}$ let $u_i \in ]0,1[^{\deg(c(y_i))}$ be the unique element such that $[y_i,\frac{1}{2}] = f([c(y_i),u_i])$. Since $c(y_i)\in  x_{\sharp}((\llbracket 0,1 \rrbracket\setminus \{1\})^{\otimes m})$, there exists an element $v_i \in [0,1[^m$ such that $[c(y_i),u_i] = [x,v_i]$ and hence $[y_i,\frac{1}{2}] = f([x,v_i]) = f\circ |x_{\sharp}|(v_i) = |\beta|\circ \phi (v_i)$. By \ref{minmaxcoord}, $\phi(v_i) \in [0,k_1[ \times \cdots \times [0,k_m[$. By \ref{bbu}, it follows that there exist elements $z_i \in (\llbracket 0,{k_1} \rrbracket \setminus \{k_1\})\otimes \cdots \otimes (\llbracket 0,{k_m} \rrbracket \setminus \{k_m\})$ and $w_i \in ]0,1[^{\deg(z_i)}$ such that $\phi(v_i) = [z_i,w_i]$. We have $[\beta(z_i),w_i]= |\beta|[z_i,w_i] = |\beta |\circ \phi (v_i) = [y_i, \frac{1}{2}]$ and hence $\deg(z_i) = 1$, $\beta (z_i) = y_i$ and $w_i = \frac{1}{2}$. Since $z_i \in (\llbracket 0,{k_1} \rrbracket \setminus \{k_1\})\otimes \cdots \otimes (\llbracket 0,{k_m} \rrbracket \setminus \{k_m\})$, also $d_1^0z_i \in (\llbracket 0,{k_1} \rrbracket \setminus \{k_1\})\otimes \cdots \otimes (\llbracket 0,{k_m} \rrbracket \setminus \{k_m\})$. For $i < l$, $c(d_1^ 1z_i) = c(z_i)$ because otherwise there would exist an index $q$ such that $c(d_1^ 1z_i) = d_q^ 1c(z_i)$. Lemma \ref{cnat} would then imply that $c(d_1^1y_i) = c(\beta(d_1^1z_i)) =  x_{\sharp}(c(d_1^1z_i)) = x_{\sharp}(d_q^ 1c(z_i)) = d_q^ 1x_{\sharp}(c(z_i)) = d_q^ 1c(\beta(z_i)) = d_q^1c(y_i)$, but $c(d_1^1y_i) = c(y_i)$. By lemma \ref{bcb}, since $z_i \in (\llbracket 0,{k_1} \rrbracket \setminus \{k_1\})\otimes \cdots \otimes (\llbracket 0,{k_m} \rrbracket \setminus \{k_m\})$,  $c(z_i) \in (\llbracket 0,1 \rrbracket\setminus \{1\})^ {\otimes m}$. Hence for $i<l$, $c(d_1^ 1z_i) \in (\llbracket 0,1 \rrbracket\setminus \{1\})^ {\otimes m}$. This implies, again by lemma \ref{bcb}, that, for $i < l$, $d_1^1z_i \in (\llbracket 0,{k_1} \rrbracket \setminus \{k_1\})\otimes \cdots \otimes (\llbracket 0,{k_m} \rrbracket \setminus \{k_m\})$. By lemma \ref{betainj},  the restriction of $\beta$ to the set $(\llbracket 0,{k_1} \rrbracket\setminus \{k_1\})\otimes \cdots \otimes (\llbracket 0,{k_m} \rrbracket\setminus \{k_m\})$ is injective. For $i < l$, $\beta (d_1^1z_i) = d_1^1y_i = d_1^ 0y_{i+1} = \beta(d_1^0z_{i+1})$ and therefore $d_1^1z_i = d_1^0z_{i+1}$. Thus, $\theta = z_{1\sharp}\cdots z_{l\sharp}$ is a well-defined path in $\llbracket 0,{k_1} \rrbracket \otimes \cdots \otimes \llbracket 0,{k_m} \rrbracket$. By construction, $|\beta|^{\mathbb I} (\theta) = \beta \circ \theta = \nu$. We have $d_1^0y_1 = \nu(0) = f_0(v)$ and hence $c(d_1^0y_1) = v$. It follows that $v \in x_{\sharp}((\llbracket 0,1 \rrbracket\setminus \{1\})^{\otimes m})$ and hence that $v = d_1^0\cdots d_1^0x$. We have  $\beta (0,\dots,0) = |\beta|_0\circ \phi_0 (0,\dots ,0) =f_0\circ |x_{\sharp}|_0 (0, \dots ,0) = f_0(d_1^0 \cdots d_1^0x) =   f_0(v) = \nu(0) = \beta (\theta (0))$ and therefore $\theta (0) = (0,\dots ,0)$.

By \ref{cnat}, $x_{\sharp}(c(z_l)) = c(\beta (z_l)) = c (y_l) = x$. Hence $c(z_l) = \iota_m$. We have $c(d_1^1z_l) = d_r^1c(z_l) = d_r^1\iota_m$ for some $r \in \{1,\dots ,m\}$. Otherwise we would have $c(d_1^1z_l) = c(z_l) = \iota_m$ and hence $c(d_1^1y_l) = c(\beta(d_1^ 1z_l)) = x_{\sharp}(c(d_1^ 1z_l)) = x_{\sharp}(\iota_m) = x$, which is impossible because $\deg(c(d_1^1y_l)) < \deg(x)$. By proposition  \ref{orderhomeo}, there exists an index $s \in \{1, \dots ,m\}$ such that $$\phi(\delta_r^1([0,1]^{m-1})) = [0,k_1]\times \cdots \times [0,k_{s-1}]\times \{k_s\} \times [0,k_{s+1}]\times \cdots \times [0,k_m].$$ Since $\phi^{-1}([d_1^1z_l,()]) \in \delta_r^1([0,1]^{m-1})$, we have $$d_1^1z_l \in \llbracket 0,{k_1} \rrbracket\otimes \cdots \otimes \llbracket 0,{k_{s-1}} \rrbracket\otimes  \{k_s\}\otimes \llbracket 0,{k_{s+1}} \rrbracket \otimes \cdots \otimes \llbracket 0,{k_m} \rrbracket .$$
Consider the path $\phi^{\mathbb I}((\hat d_r\iota_m)_{\sharp})$ in $\llbracket 0,{k_1} \rrbracket\otimes \cdots \otimes \llbracket 0,{k_m} \rrbracket$. Set $p = \length (\phi^{\mathbb I}((\hat d_r\iota_m)_{\sharp}))$. We have $\phi^{\mathbb I}((\hat d_r\iota_m)_{\sharp})(0) = \phi_0((\hat d_r\iota_m)_{\sharp}(0)) = \phi_0(0,\dots, 0) = (0,\dots ,0)$. Moreover,  $\phi^{\mathbb I}((\hat d_r\iota_m)_{\sharp})(p) = \phi_0((\hat d_r\iota_m)_{\sharp}(1)) =  \phi_0(d_1^0 \cdots d_1^0d_r^1\iota_m) = (0,\dots, 0,k_s,0, \dots ,0)$. Let $\rho$ be any path in $\llbracket 0,{k_1} \rrbracket\otimes \cdots \otimes \llbracket 0,{k_m} \rrbracket$ such that $\rho (0) = (0,\dots, 0,k_s,0, \dots ,0)$ and $\rho (\length(\rho)) = d_1^1z_l$. By \ref{detcube}, $\theta \sim \phi^{\mathbb I}((\hat d_r\iota_m)_{\sharp})\cdot \rho$. Set $\xi = |\beta|^{\mathbb I}(\rho) = \beta \circ \rho$. We  have $\nu =  |\beta|^{\mathbb I}(\theta) \sim |\beta|^{\mathbb I}(\phi^{\mathbb I}((\hat d_r\iota_m)_{\sharp})\cdot \nolinebreak \rho) = (|\beta|\circ \phi)^{\mathbb I}((\hat d_r\iota_m)_{\sharp}) \cdot \xi = (f \circ |x_{\sharp}|)^{\mathbb I}((\hat d_r\iota_m)_{\sharp})\cdot \xi = f^{\mathbb I}(|x_{\sharp}|^{\mathbb I}((\hat d_r\iota_m)_{\sharp}))\cdot \xi =  f^{\mathbb I}((\hat d_rx)_{\sharp})\cdot \xi$.  
\end{proof}

\subsection{Maximal and minimal elements}
Consider again a weak morphism of precubical sets $f\colon |P| \to |P'|$ that is a homeomorphism.

\begin{prop} \label{fminmax}
We have $f_0(M(P)) = M(P')$ and $f_0(m(P)) = m(P')$.
\end{prop}

\begin{proof}
Consider an element $v \in P_0$ such that $f_0(v)$ is not maximal in $P'$. Then there exists an element $y \in P'_1$ such that $d_1^0y = f_0(v)$. This implies that $c(y) \in P_1$ and that $d_1^0c(y) = c(d_1^0y) = v$. Thus, $v$ is not maximal in $P$. It follows that $f_0(M(P)) \subseteq M(P')$. In order to establish the reverse inclusion, consider a maximal element $v' \in P'_0$. We show first that $c(v')$ is a vertex in $P$. Suppose that this is not the case and that $\deg(c(v')) = m >0$. Let $\beta\colon \llbracket 0,{k_1} \rrbracket\otimes \cdots \otimes \llbracket 0,{k_m} \rrbracket \to P'$ be the unique morphism of precubical sets and $\phi\colon |\llbracket 0,1 \rrbracket^{\otimes m}| \to |\llbracket 0,{k_1} \rrbracket\otimes \cdots \otimes \llbracket 0,{k_m} \rrbracket|$ be the unique dihomeomorphism such that $f\circ |c(v')_{\sharp}| = |\beta| \circ \phi$. Let $u \in ]0,1[^m$ be the uniquely determined  element such that $f([c(v'),u]) = [v',()]$. Consider elements $z \in \llbracket 0,{k_1} \rrbracket\otimes \cdots \otimes \llbracket 0,{k_m} \rrbracket$ and $w \in ]0,1[^{\deg(z)}$ such that $\phi(u) = [z,w]$. Since $[\beta (z), w] = |\beta|\circ \phi (u) = f\circ |c(v')_{\sharp}| (u) = f([c(v'),u]) = [v',()]$, we have $\deg (z) = 0$, $w = ()$ and $\beta (z) = v'$. Since $u \in ]0,1[^m$, $[z,()] = \phi(u) \in ]0,k_1[ \times \cdots \times ]0,k_m[$. Therefore $z \not= (k_1, \dots, k_m)$ and there exists a $1$-cube $y \in \llbracket 0,{k_1} \rrbracket\otimes \cdots \otimes \llbracket 0,{k_m} \rrbracket$ such that $d_1^0y = z$. It follows that $d_1^0\beta (y) = v'$, which is impossible. Therefore $c(v')$ is a vertex in $P$. We have $f_0(c(v')) = v'$. We show that $c(v') \in M(P)$. Suppose that this is not the case. Then there exists an element $x \in P_1$ such that $c(v') = d_1^0x$. Since $f^{\mathbb I}(x_{\sharp}) (0) = f_0(c(v')) = v'$, there exists a path in $P'$ of non-zero length beginning in $v'$. This contradicts the hypothesis that $v'$ is maximal in $P'$. Thus, $c(v') \in M(P)$. It follows that $M(P') \subseteq f_0(M(P))$ and hence that $f_0(M(P)) = M(P')$.

An analogous argument shows that $f_0(m(P)) = m(P')$.
\end{proof}

\subsection{Homeomorphic abstraction and trace equivalence}

Let $\A = (P,I,F,\lambda)$ and $\B = (P',I',F',\lambda')$ be two $M$-HDAs.

\begin{theor}
If $\A$ is weakly regular and $\A \stackrel{\approx}{\to} \B$, then $\A\stackrel{\sim_t}{\to} \B$.
\end{theor}

\begin{proof}
Let $f\colon |P| \to |P'|$ be a weak morphism from $\A$ to $\B$ that is a homeomorphism and that satisfies $f_0(I) = I'$ and $f_0(F) = F'$. By \ref{fminmax}, $f_0(M(P)) = M(P')$ and $f_0(m(P)) = m(P')$. It follows that the functor $f_*\colon \vec \pi_1(P) \to \vec \pi_1(P')$ restricts to a functor $f_*\colon TC(\A) \to TC(\B)$ that is a bijection on objects. Consider vertices $v, w \in I \cup F \cup m(P) \cup M(P)$. By \ref{retraction}, \ref{gammac} and \ref{fJgammac}, the map $TC(\B)(f_0(v),f_0(w)) \to TC(\A)(v,w), [\omega] \mapsto [\gamma\circ c(\omega)]$ is a well-defined inverse of the map $f_*\colon TC(\A)(v,w) \to TC(\B)(f_0(v),f_0(w))$. Therefore the functor $f_*\colon TC(\A) \to TC(\B)$ is an isomorphism. Thus, $f$ is a trace equivalence from $\A$ to $\B$.
\end{proof}

\end{document}